\documentclass{article}
\usepackage[utf8]{inputenc}
\newcommand{\textcite}{\cite}

\title{Limitations on the Achievable Repair Bandwidth of Piggybacking Codes with Low Substriping}
\author{Reyna Hulett\thanks{
Computer Science Department, Stanford University.  \texttt{rmhulett@stanford.edu}.
RH's research supported in part by a NSF Graduate Research Fellowship under grant DGE-1656518. }
\ and Mary Wootters\thanks{
Computer Science and Electrical Engineering Departments, Stanford University.  \texttt{marykw@stanford.edu}.
RH and MW's research supported in part by NSF grant CCF-1657049.}
}
\date{June 2017}

\usepackage[margin=1in]{geometry}
\usepackage{blkarray}
\usepackage{amsmath}
\usepackage{amsfonts}
\usepackage{amsthm}
\usepackage{amssymb}
\usepackage{bm}
\usepackage{multirow}
\usepackage{thm-restate}
\usepackage{xcolor}
\usepackage{tikz}
\usepackage{comment}

\newtheorem{theorem}{Theorem}
\newtheorem{lemma}[theorem]{Lemma}
\newtheorem{corollary}[theorem]{Corollary}

\theoremstyle{definition}
\newtheorem{definition}{Definition}

\newtheorem{remark}[theorem]{Remark}
\newtheorem{observation}[theorem]{Observation}

\newcommand*{\vertbar}{\rule[-.5ex]{0.5pt}{2ex}}
\newcommand*{\horzbar}{\rule[.5ex]{2ex}{0.5pt}}

\makeatletter
\DeclareRobustCommand{\rvdots}{%
  \vbox{
    \baselineskip3\p@\lineskiplimit\z@
    \kern-\p@
    \hbox{.}\hbox{.}\hbox{.}
  }}
\makeatother

\newcommand{\colof}[3]{\lowercase{#1}^{(#2)}_{\bullet, #3}}
\newcommand{\rowof}[3]{\lowercase{#1}^{(#2)}_{#3, \bullet}}
\newcommand{\entryof}[4]{\lowercase{#1}^{(#2)}_{#3, #4}}
\newcommand{\entryofvec}[3]{\lowercase{#1}^{(#2)}_{#3}}
\newcommand{\Pj}[1]{P^{(#1)}}
\newcommand{\Pij}[2]{P^{(#1,#2)}}
\newcommand{\entryofF}[2]{f_{#1, #2}}
\newcommand{\cC}{\mathcal{C}}
\newcommand{\F}{\mathbb{F}}

\newcommand{\mkw}[1]{\textcolor{red}{\textbf{[#1 --mary]}}}
\begin{document}

\maketitle




\begin{abstract}
The \em piggybacking \em framework for designing erasure codes for distributed storage has empirically proven to be very useful, and has been used to design codes with desirable properties, such as low repair bandwidth and complexity. However, the theoretical properties of this framework remain largely unexplored.  We address this by adapting a general characterization of repair schemes (previously used for Reed Solomon codes) to analyze piggybacking codes with low substriping.  With this characterization, we establish a separation between piggybacking and general erasure codes, and several impossibility results for subcategories of piggybacking codes; for certain parameters, we also present explicit, optimal constructions of piggybacking codes.
\end{abstract}

\section{Introduction}
The modern world is practically overwhelmed with data, much of which is kept in large-scale distributed storage systems. These systems store large files across a number of servers, or \emph{nodes}. Due to the scale of such systems, node failure is an everyday occurrence, and the system must be robust to such failures. One way to achieve robustness is by replicating the data.  However, this clearly has high storage overhead.  Erasure coding can achieve the same reliability as replication with far less overhead.

Recently, there has been extensive effort in the field of \em coding for distributed storage \em to 
design erasure codes with desirable properties.  Two important desirable properties include an optimal reliability-redundancy trade-off, and bandwidth-efficient repair of failed nodes. 
Since the work of~\cite{dimakis2010}, there have been several constructions of \em regenerating codes \em which aim to achieve these properties.  We refer the reader to~\cite{wiki} for much more information on coding for distributed storage.

In recent work,
\textcite{rashmi2013} introduced a new \emph{piggybacking} design framework to construct such codes.  This framework modifies a base code to improve its repair properties. This framework has been employed several times to design new codes \cite{kumar2015} \cite{rashmi2013} \cite{shangguan2016} \cite{yang2015}, including one code that is being implemented in the Hadoop Distributed File System \cite{rashmi2013}.

Although the piggybacking framework has clearly been productive in practice, there has not been much theoretical analysis of its possibilities and limitations. That is the subject of this paper.

\subsection{Our Contributions}
We build on a framework introduced by \textcite{guruswami2016} for characterizing and analyzing erasure code repair schemes, which gives us a characterization of piggybacking code repair schemes in particular. This allows us to prove various impossibility results for piggybacking schemes, and to design schemes with optimal repair bandwidth for certain parameters. Specifically, our contributions are the following.

\begin{enumerate}
\item \textbf{Extension of the framework of \cite{guruswami2016}.} We adapt the characterization of repair schemes by \cite{guruswami2016}, originally introduced for Reed Solomon codes, to our setting.  More precisely, their scheme works for scalar MDS codes over finite fields, while piggybacking codes are not scalar.  We modify their approach to obtain a characterization of linear repair schemes for MDS array codes.  We specialize this to piggybacking codes for our main results, but the general framework may be of broader interest.
\item \textbf{Separation between piggybacking and general erasure codes.} Using this framework, we demonstrate that for certain parameter regimes, piggybacking cannot achieve the optimal repair bandwidth achievable by general erasure codes. Thus piggybacking is strictly less powerful than general erasure codes.
\item \textbf{Other bounds.} We additionally utilize this framework to give some limited lower bounds for piggybacking in other settings, as well as upper bounds and explicit code constructions for some specific parameters. Certain of these bounds suggest approaches to using the piggybacking design framework which may improve the attainable repair bandwidth, compared to existing practices.
\end{enumerate}

\subsection{Organization}
In Section~\ref{sec:setup} we set up notation and definitions. In Section~\ref{sec:related} we survey related work and restate our results in more detail. We introduce our characterization of repair schemes in Section~\ref{sec:char} and leverage it to prove a variety of useful lemmas. The main results for piggybacking are given in Sections~\ref{sec:any} and \ref{sec:some}. In Section~\ref{sec:concl} we conclude with some open questions.


\section{Setup and Preliminaries}\label{sec:setup}
\subsection{Notation}
In general, we will use (parenthetical) superscripts to denote different matrices and subscripts to index within a matrix. For indexing into a matrix $M^{(\ell)}$, the entry in row $i$, column $j$ will be denoted $\entryof{M}{\ell}{i}{j}$, the $i^{th}$ row will be denoted $\rowof{M}{\ell}{i}$, and the $j^{th}$ column will be denoted $\colof{M}{\ell}{j}$.

Vectors generated by indexing into a matrix will be rows or columns corresponding to their orientation in the matrix (e.g., $\rowof{M}{\ell}{i}$ is a row vector but $\colof{M}{\ell}{j}$ is a column vector). Other vectors will be considered row vectors by default. They will be typeset in bold as $\bm{v^{(\ell)}} = [v^{(\ell)}_0 \ v^{(\ell)}_1 \ \cdots]$ or $\bm{v} = [\horzbar \bm{v_0} \horzbar \ \horzbar \bm{v_1} \horzbar \ \cdots]$ if the elements are themselves vectors.

\subsection{Erasure Coding and the Exact Repair Problem}
In this paper, we restrict our focus to linear, maximum distance separable codes with linear repair schemes.  We first briefly recall some definitions.  A \em code \em $\mathcal{C}$ over an alphabet $A$ is a subset of $A^n$; if the code has size $|\mathcal{C}| = |A|^k$, we say that the \em dimension \em of $\mathcal{C}$ is $k$.  We say that such a code has the Maximum Distance Separable property (MDS property) if 
any $k$ symbols of a codeword $c \in \cC$ can determine $c$.

If the alphabet $A$ is a field, $A = \F_q$, and if $\mathcal{C} \subset \F_q^n$ is a linear subspace of $\F_q^n$, then we say $\mathcal{C}$ is \em linear. \em
A linear code $\cC$ can always be written as the image of a \em generator matrix \em $F \in \F_q^{k \times n}$; given a \em message \em $\bm{a} \in \F_q^k$, the corresponding \em codeword \em is $\bm{a}F$.
If $\mathcal{C}$ additionally has the MDS property---equivalently, if a generator matrix $F$ has the property that any $k$ columns are linearly independent---we say $\cC$ is an MDS code.
We say that a code (along with an encoding map from messages to codewords) is \em systematic \em if the $k$ symbols of the message appear as symbols of the codeword.  Notice that any linear code can be made systematic by performing row operations on the generator matrix to obtain a generator matrix so that the first $k$ columns form the identity.  In this case, the first $k$ symbols are called \em systematic \em while the remaining symbols are called \em parity symbols. \em

%

We will study \em array codes, \em where the alphabet $A$ is in fact a vector space $A = \F_q^t$.  These codes are not linear (indeed, it does not make sense for a code to be linear over $\F_q^t$), but we will study codes that are $\F_q$-linear.
\begin{definition}\label{def:arraycode}
An \emph{array code with $t$ substripes} over an alphabet $\F_q$ is a code $\cC \subset (\F_q^t)^n$ over $\F_q^t$.
We say that $\cC$ has \em linear substripes \em if $\cC$ is closed under $\F_q$-linear operations; that is, for any $\bm{c}, \bm{c'} \in \cC$ and for any $\lambda \in \F_q$, we have $\bm{c} + \lambda \bm{c'} \in \cC$.  If $\cC$ has the MDS property, we say that it is an \em MDS array code. \em
\end{definition}
We will often think of codewords of an array code as matrices $C \in \F_q^{n \times t}$, rather than vectors $\bm{c} \in (\F_q^t)^n$, and we will write $\cC \subset \F_q^{n \times t}$.
Notice that this orientation is at odds with our convention that vectors $\bm{c}$ of length $n$ are row vectors, but we will stick by it because it will be more convenient and intuitive for the diagrams in the rest of the paper.

In coding for distributed storage, the message $\bm{a}$ corresponds to a file to be stored, and the corresponding codeword $\bm{c} \in \cC$ captures how the data should be stored on the $n$ nodes: node $i$ holds the symbol $c_i$.  
In this setting, we would always like to tolerate as many node failures as possible, which means that we demand that the code $\cC$ have the MDS property.  Moreover, there are certain operations we would like to be efficient.
First, we would like to be able to recover the original message (the stored file) efficiently.  This can always be done directly if the code is systematic.   Second, while we would like to be able to handle $n-k$ failures in the worst case, a much more common scenario in many systems is a single failure~\cite{rashmi2013b}.  Thus, we would like to be able to repair a single failed node as efficiently as possible.  In this work, the measure of efficiency we consider is the \em repair bandwidth, \em which measures how much data must be downloaded to repair a single failure.

Formally, let $\cC$ be an MDS array code over $\F_q$ with $t$ substripes.  If node $i^*$ fails, then a \em repair scheme \em to repair $i^*$ using a \em repair set \em $S \subset \{0, \ldots, n-1 \} \setminus \{i^*\}$ is a collection of functions\footnote{In this work, we will only consider repair bandwidth, rather than disk access, so we allow the nodes to do arbitrary local computation.} $g_i:\F_q^t \to \F_q^{b_i}$ so that for all $\bm{c} \in \cC$, $c_{i^*}$ can be determined from $\{ g_i( c_i ) \,\mid\, i \in S \}$.  
If $\cC$ is a linear MDS array code, and if the functions $g_i$ and the method of determining $c_{i^*}$ is linear, we say that the repair scheme is \em linear. \em 

The above defines a repair scheme for a particular node $i^*$ and a particular repair set $S$.
A (linear) repair scheme with \emph{locality} $d$ for an MDS array code $\cC$ over $\F_q$ consists of (linear) repair schemes with repair sets $S$ of size $d$ for every possible failed node $i^*$.  There are two important regimes.  In the ``any $d$" regime, there must be a valid repair scheme for any repair set $S$ of size $d$.  On the other hand, in the ``some $d$" regime, we require only one valid repair set of size $d$ per possible failed node.

The \em bandwidth \em of a repair scheme for an MDS array code $\cC$ over $\F_q$ is the number of symbols of $\F_q$ needed to repair any symbol $i^*$.  In the language above, it is the maximum, over all $i^*$ and all repair sets $S$ in the scheme, of $\sum_{i \in S} b_i$.  
The \em exact repair problem \em is the problem of minimizing the repair bandwidth.  There have been several solutions proposed in the literature since the problem was introduced in~\cite{dimakis2010}.  In this work, we focus on the \em piggybacking framework, \em which we discuss in the next section.

\subsection{Piggybacking}\label{sec:pb}
In this paper, we study the piggybacking framework introduced by \textcite{rashmi2013}, (with a few assumptions, discussed below). 

A \emph{piggybacking code} $\cC$ over $\F_q$ with $t$ substripes is constructed from a ``base code" $\cC_0 \subset \F_q^n$ and $\binom{t}{2}$ ``piggybacking functions" $p^{(i,j)}: \F_q^k \to \F_q^n$.  For this work, we assume that the base code $\cC_0$ is a (scalar) MDS code over $\F_q$; in particular, it is linear, with a generator matrix $F \in \F_q^{k \times n}$.  We also assume that the piggybacking functions $p^{(i,j)}$ are linear; in particular, they can be represented by matrices $P^{(i,j)} \in \F_q^{k \times n}$.  

With these assumptions, we define a piggybacking code (with a scalar MDS base code) over $\F_q$ as follows.

%
\begin{definition}\label{def:pb}
Let $F \in \F_q^{k \times n}$ be the generator matrix of an MDS code $\cC_0$, and take a collection of piggybacking matrices
\[ \{ \Pij{i}{j} \,\mid\, i \in [0,t-2], j \in [i+1,t-1] \} \subset \mathbb{F}_q^{k \times n}. \]
Consider the MDS array code $\cC$ over $\F_q$ with $t$ linear substripes, defined as follows.
Given a message $\bm{a} \in \F_q^{kt}$ given by
\[\bm{a} = [\horzbar \bm{a_0} \horzbar \ \cdots \ \horzbar \bm{a_{t-1}} \horzbar] \]
(where each $\bm{a_i} \in \F_q^k$), we form a codeword $C \in \F_q^{n \times t}$ so that the $i^{th}$ substripe is
\[ c_{\bullet,i} = (\bm{a_0} \Pij{0}{i} + \cdots + \bm{a_{i-1}} \Pij{i-1}{i} + \bm{a_i} F)^T. \]

We say that $\cC$ is a
\emph{$(n,k)$ piggybacking code with a scalar MDS base code} (henceforth a \emph{piggybacking code}) with $t$ substripes over $\mathbb{F}_q$.
\end{definition}

We illustrate a piggybacking code formed from $F$ and $\{P^{(i,j)}\}$ below.

\begin{center}
\begin{tabular}{c|c|c|c|c} 
 & substripe 0 & substripe 1 & $\cdots$ & substripe t-1 \\
\hline
node 0 & \vertbar & \vertbar &  & \vertbar \\
$\rvdots$ & $\left(\bm{a_0} F\right)^T$ & $\left(\bm{a_0} \Pij{0}{1} + \bm{a_1} F\right)^T$ & $\cdots$ & $\left(\bm{a_0} \Pij{0}{t-1} + \cdots + \bm{a_{t-2}} \Pij{t-2}{t-1} + \bm{a_{t-1}} F\right)^T$ \\
node n-1 & \vertbar & \vertbar &  & \vertbar
\end{tabular}
\end{center}
As with all MDS array codes, we will represent codewords as matrices in $\F_q^{n \times t}$ of the form
\[\begin{bmatrix}
\vertbar  & \vertbar						& 			& \vertbar	\\
(\bm{a_0} F)^T & \left(\bm{a_0} \Pij{0}{1} + \bm{a_1} F\right)^T		& \cdots	& \left(\bm{a_0} \Pij{0}{t-1} + \cdots + \bm{a_{t-2}} \Pij{t-2}{t-1} + \bm{a_{t-1}} F\right)^T	\\
\vertbar & \vertbar						& 			& \vertbar	\\
\end{bmatrix}\]

As noted in \cite{rashmi2013}, piggybacking codes using an MDS base code remain MDS, but may have improved repair properties; in particular, they may have reduced repair bandwidth.

In addition to general piggybacking codes, we will also consider a subcategory of codes which only piggyback in the last substripe of each node, inspired by the approach of \textcite{yang2015}. We dub these \emph{linebacking codes}.
\begin{definition}\label{def:lb}
An \emph{$(n,k)$ linebacking code with a scalar MDS base code} (henceforth a \emph{linebacking code}) $\mathcal{C}$ over the finite field $\mathbb{F}_q$ with $t$ substripes 
is a $(n,k)$ piggybacking code, with the additional property that
all piggybacking matrices $\Pij{i}{j}$ such that $j \neq t-1$ are zero. Thus we drop the index $j$ indicating which substripe the piggyback is added to and denote $\Pij{i}{t-1}$ by $\Pj{i}$. 
\end{definition}

Thus for message $\bm{a} = [\horzbar \bm{a_0} \horzbar \ \cdots \ \horzbar \bm{a_{t-1}} \horzbar]$, a linebacking code stores
\begin{center}
\begin{tabular}{c|c|c|c|c|c} 
 & substripe 0 & substripe 1 & $\cdots$ & substripe t-2 & substripe t-1 \\
\hline
node 0 & \vertbar & \vertbar &  & \vertbar & \vertbar \\
$\rvdots$ & $(\bm{a_0} F)^T$ & $(\bm{a_1} F)^T$ & $\cdots$ & $(\bm{a_{t-2}} F)^T$ & $\left(\bm{a_0} \Pj{0} + \cdots + \bm{a_{t-2}} \Pj{t-2} + \bm{a_{t-1}} F\right)^T$ \\
node n-1 & \vertbar & \vertbar &  & \vertbar & \vertbar
\end{tabular}
\end{center}
and has codewords of the form
\[\begin{bmatrix}
\vertbar & \vertbar	& 			& \vertbar	 	& \vertbar	\\
(\bm{a_0} F)^T  & (\bm{a_1} F)^T	& \cdots	& (\bm{a_{t-2}} F)^T	& \left(\bm{a_0} \Pij{0}{t-1} + \cdots + \bm{a_{t-2}} \Pij{t-2}{t-1} + \bm{a_{t-1}} F\right)^T	\\
\vertbar & \vertbar	& 			& \vertbar		& \vertbar	\\
\end{bmatrix}\]

It is not hard to see that, as with general linear MDS codes, piggybacking (respectively, linebacking) codes can be made systematic via a remapping of message symbols, while retaining their piggybacking (linebacking) structure. Thus as with general codes, we can assume without loss of generality that the first $k$ nodes of a piggybacking (linebacking) code are systematic, meaning substripe $i$ stores message chunk $\bm{a_i}$ on the first $k$ nodes. 



\section{Related Work and Our Results}\label{sec:related}

The piggybacking framework for designing error correcting codes for distributed storage was introduced by \textcite{rashmi2013}. It is as described in Definition~\ref{def:pb}, except that we have made the following assumptions.  
First, we assume that the piggybacking functions are linear---in general this is not required---and second, that the base code is a scalar MDS code---in general, the base code may itself be an MDS array code. However, we note that all piggybacking codes in the literature do use linear piggybacking functions \cite{kumar2015} \cite{rashmi2013} \cite{shangguan2016} \cite{yang2015}. Almost all use scalar MDS base codes as well, except \cite{yang2015} and one of four constructions in \cite{rashmi2013}, which are specifically designed to improve the repair properties of parity nodes for existing array codes.

Furthermore, in \cite{rashmi2013}, an invertible linear tranformation may be applied to the data stored on each node in order to reduce the data-read.  However, since in this work we are only concerned with repair bandwidth, this does not matter for us and we omit it from Definition~\ref{def:pb}.

The piggybacking design framework has been used to produce codes with low data-read and bandwidth for repairing individual failed nodes. \textcite{rashmi2013} used the framework to design explicit codes with the lowest data-read and bandwidth among known solutions for a few specific settings, including (high-rate) MDS codes with low substriping, the domain of interest in this paper. Extending their ideas, \textcite{yang2015} showed how to modify codes with optimal repair bandwidth for systematic nodes to use piggybacking to obtain asymptotically optimal bandwidth for parity nodes as well. Interestingly, \cite{yang2015} obtained these results for linebacking codes (Definition~\ref{def:lb}), which is more restricted than general piggybacking. The piggybacking framework was also employed by \textcite{shangguan2016} to design codes with low repair complexity, and by \textcite{kumar2015} as part of a compound design using both piggybacking and simple parity checks.

However, little is understood about the theoretical possibilities and limitations of codes designed using the piggybacking framework. Nor is there much understanding of how to choose piggybacking functions to achieve desirable repair properties. Although \cite{yang2015} takes a more principled approach than others to choosing the piggybacking functions, some of their choices---including piggybacking only in the last substripe and always using all the systematic nodes in repairing a failed parity node---do not have a rigorous theoretical backing.

Here we explore the theoretical limitations on achievable repair bandwidth for piggybacking codes with scalar MDS base codes and with a small number of substripes $t \leq n - k$.  As in Definition~\ref{def:arraycode}, we do not allow for \em symbol extension; \em that is, we treat the elements of $\F_q$ as indivisible and measure bandwidth in units of symbols of $\F_q$.  We focus primarily on the regime where any failed node must be repairable from \emph{any} set of $d$ other nodes, where $d$ is the locality. The alternative is that for each failed node, there must exist \emph{some} set of $d$ other nodes which repair the failed node. This alternative regime is less restrictive, and the achievable repair bandwidth is less well characterized. While both regimes have been studied in the literature \cite{shah2012}, the piggybacking design framework has primarily been employed in the latter, less restrictive regime \cite{kumar2015} \cite{rashmi2013} \cite{shangguan2016} \cite{yang2015}.

\paragraph{Known lower bounds.}
In any setting, for MDS codes, we have the \em cut-set bound \em on the repair bandwidth, which states that we must download $b \geq \frac{td}{d-k+1}$ symbols~\cite{dimakis2010}. Since this is decreasing in $d$, we can also set $d$ to the maximum/optimal value $n-1$ to get the bound $b \geq t \frac{n-1}{n-k}$. However, if $t < n-k$ and in the absence of symbol extension---the setting we consider here---this is not achievable, since it would require downloading less than a full symbol from each node. Since we must download at least one full symbol from every participating node, we can say $b \geq d$ which gives the bound
\[b \geq \frac{td}{d-k+1} \geq \frac{tb}{b-k+1}\]
which implies
\[b \geq k + t - 1,\]
which also matches the trivial lower bound for any MDS code \cite{guruswami2016}. We will call $b = k+t-1$ ``perfect bandwidth.''\footnote{Throughout we will assume $2 \leq t \leq n-k$ and $2 \leq k$. Otherwise, if $t = 1$, there is no piggybacking and any MDS code achieves perfect bandwidth; if $t > n-k$, achieving perfect bandwidth is impossible; and if $k =1$, the straightforward lower bound on bandwidth $k+t-1$ and the trivially achievable bandwidth for MDS codes $kt$ are equal.}

Note that this bound cannot be tight for large $t > n-k$. However, for $t \leq n-k$, and in the regime where any $d$ nodes must be able to repair the failed node, \textcite{shah2012} demonstrated that perfect bandwidth is achievable for $t > k-3$, but cannot be achieved by a linear code without symbol extension for $t \leq k-3$. Under the weaker requirement that only \emph{some} set of $d$ nodes must be able to repair the failed node, \textcite{suh2011} showed that the cut-set bound (and thus perfect bandwidth) is achievable provided $k \leq \max\{\frac{n}{2}, 3\}$ and $d \geq 2k-1$ which in our setting translates to substriping $t \geq k$, with field size at most $2(n-k)$. Although the cut-set bound has been shown to be achievable for large $t,q$ and general $n,k$ \cite{rashmi2011} \cite{ye2017}, to our knowledge the question of achieving perfect bandwidth when $t \leq n-k$ in general remains open.

\subsection{Our Results}

In this paper, we study the ability of piggybacking codes (with scalar MDS base codes) to achieve perfect bandwidth when $t \leq n-k$, using linear repair schemes. By adapting the framework of \cite{guruswami2016}, we give complete results for the regime where any $d$ nodes must be able to repair a failed node, and partial progress for the regime where only some $d$ nodes repair. These results, summarized in Tables~\ref{table:any} and \ref{table:some}, are as follows.

\begin{itemize}
\item ``Any $d$'' regime:
\begin{itemize}
	\item Piggybacking codes cannot achieve perfect bandwidth for $k \geq 3$, and thus are strictly weaker than general MDS codes (Theorem~\ref{thm:anyimposs}).
	\item Linebacking codes suffice to achieve perfect bandwidth for $k = 2$ with large field size $q$ (Theorem~\ref{thm:k=2}).
\end{itemize}
\item ``Some $d$'' regime:
\begin{itemize}
	\item Piggybacking codes are more powerful in this regime than the ``any $d$'' regime, as demonstrated by an example perfect bandwidth linebacking code with $k=3$ (Theorem~\ref{thm:someconstruction}).
	\item By extension from the ``any $d$'' regime, piggybacking codes do not exist for $k \geq 3, t = n-k$, and thus are strictly weaker than general MDS codes (Theorem~\ref{thm:anyimposs}), but linebacking codes do exist for $k=2$ with large field size $q$ (Theorem~\ref{thm:k=2}).
	\item Linebacking codes cannot achieve perfect bandwidth for $k\geq 3$ and $t > \frac{n-k+1}{2}$ (Theorem~\ref{thm:linebacking}).
	\item Linebacking codes which follow the common practice of using all remaining systematic nodes to repair a failed node cannot achieve perfect bandwidth for $k\geq 3$ and $t > \frac{n-k-1}{\sqrt{k}}$ (Theorem~\ref{thm:linesyst}).
	\item For $t=2,k=2$ we give an explicit construction of perfect bandwidth linebacking codes with small field size $q \geq k+1$ (Theorem~\ref{thm:somek=2}); compare to \cite{suh2011} where $q \geq 2(n-k)$.
\end{itemize}
\end{itemize}

\begin{center}
\begin{table}[h]
\centering
\begin{tabular}{|c|l|l|l|}
\hline
 & \multicolumn{1}{c|}{General} & \multicolumn{1}{c|}{Piggybacking} & \multicolumn{1}{c|}{Linebacking} \\
\hline
\multirow{2}{*}{$k\geq 3$} & $t > k-3$: exist for $q \geq n-k+t$ & do not exist (Thm.~\ref{thm:anyimposs}) & do not exist (Thm.~\ref{thm:anyimposs}) \\
 & $t \leq k-3$: do not exist \cite{shah2012} & & \\
\hline
$k=2$ & exist for $q \geq n-k+t$ \cite{shah2012} & exist for some $q$ (Thm.~\ref{thm:k=2}) & exist for some $q$ (Thm.~\ref{thm:k=2}) \\
\hline
\end{tabular}
\caption{Existence of perfect bandwidth MDS array codes when $t \leq n-k$ in the ``any $d$'' regime}\label{table:any}
\end{table}
\end{center}

\begin{center}
\begin{table}[h]
\centering
\begin{tabular}{|l|l|l|l|}
\hline
 & \multicolumn{1}{c|}{General} & \multicolumn{1}{c|}{Piggybacking} & \multicolumn{1}{c|}{Linebacking} \\
\hline
\multirow{4}{*}{$k\geq 3$} & $t \geq k$: exist for & $t=n-k$: do not exist (Thm.~\ref{thm:anyimposs}) & $t=n-k$: do not exist (Thm.~\ref{thm:anyimposs}) \\
 & \quad $k \leq \max\{\frac{n}{2}, 3\}$, & $t=2$: example $(6,3)$ construction & $t > \frac{n-k+1}{2}$: do not exist (Thm.~\ref{thm:linebacking}) \\
 & \quad $q \geq 2(n-k)$ \cite{suh2011} & \quad with $q=7$ (Thm.~\ref{thm:someconstruction}) & $t=2$: example $(6,3)$ construction \\
 & & & \quad with $q=7$ (Thm.~\ref{thm:someconstruction}) \\
\hline
\multirow{3}{*}{$k=2$} & exist for & exist for some $q$ (Thm.~\ref{thm:k=2}) & exist for some $q$ (Thm.~\ref{thm:k=2}) \\
 & \quad $q \geq 2(n-k)$ \cite{suh2011} & $t=2$: construction for & $t=2$: construction for \\
 & & \quad $q \geq k+1$ (Thm.~\ref{thm:somek=2}) & \quad $q \geq k+1$ (Thm.~\ref{thm:somek=2}) \\
\hline
\end{tabular}
\caption{Existence of perfect bandwidth MDS array codes when $t \leq n-k$ in the ``some $d$'' regime}\label{table:some}
\end{table}
\end{center}


\section{Characterization of Repair Schemes}\label{sec:char}
The work of~\cite{guruswami2016} provides a characterization of linear repair schemes for scalar MDS codes.  Their framework relies on the fact that a scalar MDS code $\cC$ is linear over its alphabet.  In our case, piggybacking codes---and more generally MDS array codes with linear substripes---are linear over $\F_q$, but not over $\F_q^t$.  (Indeed, linearity over $\F_q^t$ does not immediately make sense, as $\F_q^t$ does not have a natural notion of multiplication). 
However, the approach of~\cite{guruswami2016} still makes sense in this context.  The main reason linearity was important to the approach of \cite{guruswami2016} was because their characterization involved the \em dual code, \em $\cC^\perp$.  We may introduce a similar notion for MDS array codes.  

\begin{definition}
Let $\cC \subset \F_q^{n \times t}$ be an MDS array code with $t$ linear substripes over $\F_q$.
The \emph{dual code} of $\cC$ is $\mathcal{C}^\perp := \{X \in \mathbb{F}_q^{n \times t} \mid \langle X,C \rangle = 0 \ \forall C \in \mathcal{C} \}$, where $\langle \cdot,\cdot \rangle$ denotes the Frobenius inner product. 
\end{definition}

\begin{theorem}\label{thm:matrixchar}
Let $\mathcal{C} \subset \mathbb{F}_q^{n \times t}$ be a $(n,k)$ MDS array code with $t$ linear substripes over $\F_q$. For a fixed node $i^*$ and set of nodes $S \not\owns i^*$, the following are equivalent.
\begin{enumerate}
\item There is a linear repair scheme for node $i^*$ from $S$ with bandwidth $b$.
\item There exists a set of $t$ dual codewords, the \emph{repair matrices}, $\{W^{(0)}, W^{(1)}, \dots, W^{(t-1)}\} \subset \mathcal{C}^\perp$ such that the only non-zero rows of each $W^{(j)}$ are $i^*$ and $S$, and
\[\dim ( \{ \rowof{W}{j}{i^*} \mid j \in [0,t-1] \} ) = t \]
\[\sum_{i \neq i^*} \dim ( \{ \rowof{W}{j}{i} \mid j \in [0,t-1] \} ) \leq b \]
\end{enumerate}
\end{theorem}
See Figure~\ref{fig:repairscheme} for an illustration of a set of repair matrices.
The proof of Theorem~\ref{thm:matrixchar} follows very similar to the approach in \cite{guruswami2016}.  We include a proof in Appendix \ref{app:thm2} for completeness, but we sketch one direction in Figure~\ref{fig:alg}, showing how a set of repair matrices yields a repair scheme.

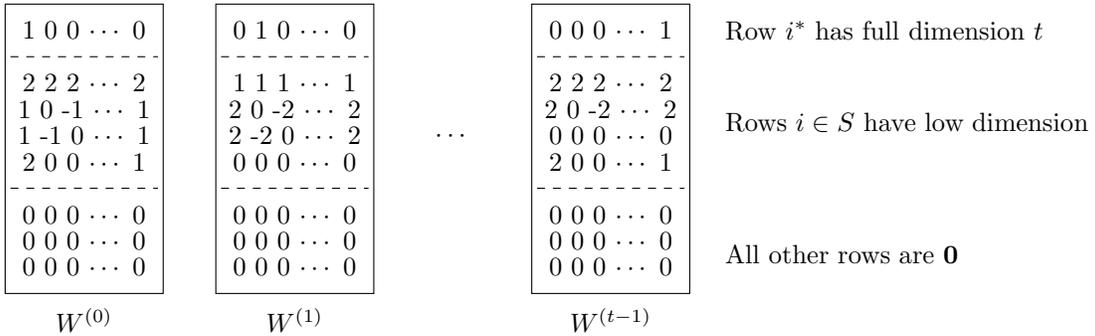
\begin{figure}[h]
\begin{center}
\begin{tikzpicture}[scale=.35]

\begin{scope}
\draw (0,-1) rectangle (6,10);
\node at (3,9) {1\ 0\ 0\ $\cdots$\ 0};
\draw[dashed] (.2,8) to (5.8,8);
\node at (3,7) {2\ 2\ 2\ $\cdots$\ 2};
\node at (3,6) {1\ 0\ -1\ $\cdots$\ 1};
\node at (3,5) {1\ -1\ 0\ $\cdots$\ 1};
\node at (3,4) {2\ 0\ 0\ $\cdots$\ 1};
\draw[dashed] (.2,3) to (5.8,3);
\node at (3,2) {0\ 0\ 0\ $\cdots$\ 0};
\node at (3,1) {0\ 0\ 0\ $\cdots$\ 0};
\node at (3,0) {0\ 0\ 0\ $\cdots$\ 0};
\node at (3,-2) {$W^{(0)}$};
\end{scope}

\begin{scope}[xshift=8cm]
\draw (0,-1) rectangle (6,10);
\node at (3,9) {0\ 1\ 0\ $\cdots$\ 0};
\draw[dashed] (.2,8) to (5.8,8);
\node at (3,7) {1\ 1\ 1\ $\cdots$\ 1};
\node at (3,6) {2\ 0\ -2\ $\cdots$\ 2};
\node at (3,5) {2\ -2\ 0\ $\cdots$\ 2};
\node at (3,4) {0\ 0\ 0\ $\cdots$\ 0};
\draw[dashed] (.2,3) to (5.8,3);
\node at (3,2) {0\ 0\ 0\ $\cdots$\ 0};
\node at (3,1) {0\ 0\ 0\ $\cdots$\ 0};
\node at (3,0) {0\ 0\ 0\ $\cdots$\ 0};
\node at (3,-2) {$W^{(1)}$};
\end{scope}

\node at (17,5) {$\cdots$};

\begin{scope}[xshift=20cm]
\draw (0,-1) rectangle (6,10);
\node at (3,9) {0\ 0\ 0\ $\cdots$\ 1};
\draw[dashed] (.2,8) to (5.8,8);
\node at (3,7) {2\ 2\ 2\ $\cdots$\ 2};
\node at (3,6) {2\ 0\ -2\ $\cdots$\ 2};
\node at (3,5) {0\ 0\ 0\ $\cdots$\ 0};
\node at (3,4) {2\ 0\ 0\ $\cdots$\ 1};
\draw[dashed] (.2,3) to (5.8,3);
\node at (3,2) {0\ 0\ 0\ $\cdots$\ 0};
\node at (3,1) {0\ 0\ 0\ $\cdots$\ 0};
\node at (3,0) {0\ 0\ 0\ $\cdots$\ 0};
\node at (3,-2) {$W^{(t-1)}$};
\end{scope}

\begin{scope}[xshift=27cm]
\node[anchor=west] at (0,9) {Row $i^*$ has full dimension $t$};
\node[anchor=west] at (0,5.5) {Rows $i \in S$ have low dimension};
\node[anchor=west] at (0,.5) {All other rows are $\bm{0}$};
\end{scope}
\end{tikzpicture}
\end{center}
\caption{An example illustrating the structure of a linear repair scheme with $q = 3$.}\label{fig:repairscheme}
\end{figure}

\begin{figure}[h]
\noindent\fbox{%
    \parbox{\textwidth}{%
\begin{tabular}{cp{.9\textwidth}c}
\vspace{.2cm} \\
\hspace{.2cm} &
\textbf{Linear repair scheme, given repair matrices.}
Suppose that $\{W^{(0)}, W^{(1)}, \dots W^{(t-1)}\}$ are a set of repair matrices for a node $i^*$ with repair set $S$, as in Theorem~\ref{thm:matrixchar}.   Let $C \in \mathcal{C}$.  Then we can define a linear repair scheme as follows.
\begin{enumerate}
\item For every node $i \neq i^*$, let $Q_i \subset \F_q^t$ be any basis of $\text{span}(\{\rowof{W}{j}{i} \mid j \in [0,t-1]\})$. 
We say that $Q_i$ is the \em query set \em for node $i$.
Observe that $Q_i = \emptyset \ \forall i \notin S$ so only nodes in the repair set will be queried. 
\item For every query vector $\bm{q} \in Q_i$, node $i$ sends $\bm{q} \cdot c_{i,\bullet}$ to the replacement node. Since $\sum_{i \neq i^*} |Q_i| \leq b$, at most $b$ symbols of $\mathbb{F}_q$ are downloaded.
\item The replacement node now has enough information to recover $\rowof{W}{j}{i^*} \cdot c_{i^*,\bullet} = \langle W^{(j)}, C\rangle - \sum_{i \neq i^*} \rowof{W}{j}{i} \cdot c_{i,\bullet}i$ for all $j$, since $W^{(j)} \in \mathcal{C}^\perp$ implies $\langle W^{(j)}, C\rangle = 0$ and $\rowof{W}{j}{i} \cdot c_{i,\bullet}$ for $i \neq i^*$ can be recovered from the responses to the query $Q_i$. Since $\dim (\{ \rowof{W}{j}{i^*} \mid j \in [0,t-1] \}) = t$, this gives $t$ linearly independent equations, and we can solve for $c_{i^*,\bullet}$, thus repairing the failed node.
\end{enumerate}
& \hspace{.2cm}
\end{tabular}
    }%
}
\caption{Turning a set of repair matrices into a linear repair scheme.  This algorithm proves one direction of Theorem~\ref{thm:matrixchar}. See Appendix~\ref{app:thm2} for the full proof.}
\label{fig:alg}
\end{figure}

Henceforth we will refer interchangeably to a linear repair scheme for $i^*$ from a set $S$, and a set of $t$ \emph{repair matrices} for $i^*, S$ as defined above. 

\subsection{Piggybacking Code Repair Schemes}\label{sec:pbchar}
Now that we have characterized repair schemes as sets of dual codewords, we can analyze the repair schemes of piggybacking codes, and those that achieve perfect bandwidth in particular. In the remainder of Section~\ref{sec:char}, we develop several lemmas using this characterization which will allow us to prove our main results in Sections~\ref{sec:any} and \ref{sec:some}.

The next lemma specializes the definition of a dual code to piggybacking codes.
\begin{lemma}\label{lem:cols}
Let $\mathcal{C}$ be an $(n,k)$ piggybacking code over $\mathbb{F}_q$ with $t$ substripes, base code $\mathcal{C}_0$ with generator matrix $F$, and piggybacking matrices $\{\Pij{i}{j} \mid i \in [0,t-2], j \in [i+1,t-1]\}$. A matrix $X \in \mathbb{F}_q^{n \times t}$ is in $\mathcal{C}^\perp$ if and only if 
\[ F x_{\bullet,i} + \Pij{i}{i+1} x_{\bullet,i+1} + \cdots + \Pij{i}{t-1} x_{\bullet,t-1} = \bm{0}^T \ \forall i \]
\end{lemma}

\begin{proof}
$X \in \mathcal{C}^\perp$ if and only if
\begin{align*}
\langle X, C \rangle &= 0 \ &\forall C \in \mathcal{C} \\
x_{\bullet,0}^T \cdot c_{\bullet,0} + \cdots + x_{\bullet,t-1}^T \cdot c_{\bullet,t-1} &= 0 \ &\forall C \in \mathcal{C} \\
x_{\bullet,0}^T \cdot \left(\bm{a_0} F\right)^T + \cdots + x_{\bullet,t-1}^T \cdot \left(\bm{a_0} \Pij{0}{t-1} + \cdots + \bm{a_{t-2}} \Pij{t-2}{t-1} + \bm{a_{t-1}} F \right)^T &= 0 \ &\forall \bm{a_0}, \dots, \bm{a_{t-1}} \in \mathbb{F}_q^k \\
\end{align*}
Since the above holds for all $\bm{a_0},\ldots,\bm{a_{t-1}}$, for an arbitrary $i$ we consider $\bm{a_j} = 0$ for $j \neq i$ and $\bm{a_i} \in \F_q^k$.  This yields the equivalent requirement
\begin{align*}
x_{\bullet,i}^T \cdot \left(\bm{a_i} F\right)^T + x_{\bullet,i+1}^T \cdot \left(\bm{a_i} \Pij{i}{i+1}\right)^T + \cdots + x_{\bullet,t-1}^T \cdot \left(\bm{a_i} \Pij{i}{t-1}\right)^T &= 0 \ &\forall i \ \forall \bm{a_i} \in \mathbb{F}_q^k \\
\bm{a_i} F x_{\bullet,i} + \bm{a_i} \Pij{i}{i+1} x_{\bullet,i+1} + \cdots + \bm{a_i} \Pij{i}{t-1} x_{\bullet,t-1} &= [0] \ &\forall i \ \forall \bm{a_i} \in \mathbb{F}_q^k \\
F x_{\bullet,i} + \Pij{i}{i+1} x_{\bullet,i+1} + \cdots + \Pij{i}{t-1} x_{\bullet,t-1} &= \bm{0}^T \ &\forall i \\
\end{align*}
as desired.
\end{proof}

Note that the restriction on the last column of a repair matrix $W$, $F w_{\bullet, t-1} = \bm{0}^T$, is equivalent to $w_{\bullet,t-1}^T \in \mathcal{C}_0^\perp$ where $\mathcal{C}_0^\perp \subset \mathbb{F}_q^n$ is the dual of the base code. Because $\mathcal{C}_0$ is a scalar MDS code, $\mathcal{C}_0^\perp$ is itself an MDS code, and thus is a linear subspace with minimum weight $k+1$.

\begin{corollary}\label{cor:cols}
Let $\mathcal{C}$ be an $(n,k)$ \emph{linebacking} code over $\mathbb{F}_q$ with $t$ substripes, base code generator matrix $F$, and piggybacking matrices $\{\Pj{i} \mid i \in [0,t-2]\}$. A matrix $X \in \mathbb{F}_q^{n \times t}$ is in $\mathcal{C}^\perp$ if and only if
\begin{align*}
F x_{\bullet,t-1} &= \bm{0}^T \\
F x_{\bullet,i} + \Pj{i} x_{\bullet,t-1} &= \bm{0}^T \ \forall i \neq t-1
\end{align*}
\end{corollary}

It will also be useful to note operations we can perform on a set of repair matrices which result in an equivalent repair scheme.
\begin{definition}
For a given MDS array code $\mathcal{C}$ with $t$ linear substripes, and failed node $i^*$, two repair schemes $\{W^{(0)}, \dots, W^{(t-1)}\}$ and $\{V^{(0)}, \dots, V^{(t-1)}\}$ are \emph{equivalent} if for every row $i$,
\[\text{span}(\{\rowof{W}{j}{i} \mid j \in [0,t-1]\}) = \text{span}(\{\rowof{V}{j}{i} \mid j \in [0,t-1]\})\]
That is, two repair schemes are equivalent if they can share the same queries $\{Q_i\}$ in the algorithm in Figure~\ref{fig:alg}. This also means the two schemes have the same repair set and repair bandwidth.
\end{definition}

\begin{lemma}\label{lem:equi}
Let $\mathcal{C}$ be an $(n,k)$ MDS array code over $\mathbb{F}_q$ with $t$ linear substripes. Consider a repair scheme $\{W^{(0)}, \dots, W^{(t-1)}\}$ for failed node $i^*$. The following operations on $\{W^{(0)}, \dots, W^{(t-1)}\}$ produce an equivalent repair scheme.
\begin{enumerate}
\item Scaling any repair matrix $W^{(j)}$ by a constant $\kappa \in \mathbb{F}_q \backslash \{0\}$
\item Adding a multiple of one repair matrix $\kappa W^{(j)}$ to another $W^{(\ell)}$ where $\kappa \in \mathbb{F}_q$, $j \neq \ell$
\end{enumerate}
\end{lemma}

\begin{proof}
First, note that the dual code $\mathcal{C}^\perp$ of a MDS array code with linear substripes is a subspace of $\F_q^{n \times t}$,
so the matrices obtained by scaling or adding two dual codewords are still dual codewords. Now, consider the effect of these operations on the row sets: $\{\rowof{W}{j}{i} \mid j \in [0,t-1]\}$ for row $i$. Scaling or adding two repair matrices is equivalent to scaling or adding two rows in each row set. Since these are elementary row operations, they do not change the subspace spanned by each row set, and thus by definition these operations result in an equivalent repair scheme.
\end{proof}

The above holds for any MDS array code with linear substripes. For piggybacking and linebacking schemes specifically, we also define a weaker notion of equivalence which permits more operations on the repair scheme.
\begin{definition}
For a given MDS array code $\mathcal{C}$ with linear substripes, and failed node $i^*$, two repair schemes $\{W^{(0)}, \dots, W^{(t-1)}\}$ and $\{V^{(0)}, \dots, V^{(t-1)}\}$ are \emph{download-equivalent} if for every row $i$,
\[\dim(\{\rowof{W}{j}{i} \mid j \in [0,t-1]\}) = \dim(\{\rowof{V}{j}{i} \mid j \in [0,t-1]\})\]
That is, two repair schemes are download-equivalent if they download the same number of symbols from each node. Two equivalent schemes are also download-equivalent, but not vice versa. Download-equivalent schemes must have the same repair set and bandwidth, but not the same queries $\{Q_i\}$.
\end{definition}

\begin{lemma}\label{lem:downloadequi}
Let $\mathcal{C}$ be an $(n,k)$ MDS array code over $\mathbb{F}_q$ with $t$ linear substripes. Consider a repair scheme $\{W^{(0)}, \dots, W^{(t-1)}\}$ for failed node $i^*$. The following operations on $\{W^{(0)}, \dots, W^{(t-1)}\}$ produce a download-equivalent repair scheme.
\begin{enumerate}
\item For a piggybacking code, adding $\kappa \colof{W}{j}{t-1}$ to $\colof{W}{j}{0}$ for some fixed $\kappa \in \mathbb{F}_q$ and every $j$
\item For a linebacking code, adding $\kappa \colof{W}{j}{t-1}$ to $\colof{W}{j}{\ell}$ for some fixed $\kappa \in \mathbb{F}_q$, $\ell \neq t-1$ and every $j$
\end{enumerate}
\end{lemma}

\begin{proof}
First, we will argue that the resulting matrices are still dual codewords.

\begin{enumerate}
\item For a piggybacking code $\mathcal{C}$, recall from Lemma~\ref{lem:cols} that $X \in \mathcal{C}^\perp$ if and only if
\begin{equation}
\label{eq:whyzero}
 F x_{\bullet,i} + \Pij{i}{i+1} x_{\bullet,i+1} + \cdots + \Pij{i}{t-1} x_{\bullet,t-1} = \bm{0}^T \ \forall i 
\end{equation}
Consider the repair matrix $W^{(j)}$ and the resulting matrix $W^{(j)'}$ obtained by adding $\kappa$ times the last column to the $0^{th}$ column. Only the $0^{th}$ column of $W^{(j)'}$ differs from $W^{(j)}$, so we need only check the equation for $i=0$:
\begin{align*}
F w^{(j)'}_{\bullet,0} + \Pij{0}{1} w^{(j)'}_{\bullet,1} + \cdots + \Pij{0}{t-1} w^{(j)'}_{\bullet,t-1} &= F (w^{(j)}_{\bullet,0} + \kappa w^{(j)}_{\bullet,t-1} ) + \Pij{0}{1} w^{(j)}_{\bullet,1} + \cdots + \Pij{0}{t-1} w^{(j)}_{\bullet,t-1} \\
&= \kappa F w^{(j)}_{\bullet,t-1} + F w^{(j)}_{\bullet,0} + \Pij{0}{1} w^{(j)}_{\bullet,1} + \cdots + \Pij{0}{t-1} w^{(j)}_{\bullet,t-1} \\
&= \kappa \bm{0}^T + \bm{0}^T = \bm{0}^T
\end{align*}
as desired, where the last line results from two applications of \eqref{eq:whyzero} to $W^{(j)} \in \cC^\perp$.

\item For a linebacking code $\mathcal{C}$, recall from Corollary~\ref{cor:cols} that $X \in \mathcal{C}^\perp$ if and only if
\begin{align*}
F x_{\bullet,t-1} &= \bm{0}^T \\
F x_{\bullet,i} + \Pj{i} x_{\bullet,t-1} &= \bm{0}^T \ \forall i \neq t-1
\end{align*}
Consider the repair matrix $W^{(j)}$ and the resulting matrix $W^{(j)'}$ obtained by adding $\kappa$ times the last column to the $\ell^{th}$ column. Only the $\ell^{th}$ column of $W^{(j)'}$ differs from $W^{(j)}$, so we need only check the equation for $i=\ell$:
\begin{align*}
F w^{(j)'}_{\bullet,\ell} + \Pj{\ell} w^{(j)'}_{\bullet,t-1} &= F (w^{(j)}_{\bullet,\ell} + \kappa w^{(j)}_{\bullet,t-1}) + \Pj{\ell} w^{(j)}_{\bullet,t-1} \\
&= \kappa F w^{(j)}_{\bullet,t-1} + F w^{(j)}_{\bullet,\ell} + \Pj{\ell} w^{(j)}_{\bullet,t-1} \\
&= \kappa \bm{0}^T + \bm{0}^T = \bm{0}^T
\end{align*}
as desired.
\end{enumerate}

Now, consider the effect of these operations on the row sets $\{\rowof{W}{j}{i} \mid j \in [0,t-1]\}$ for row $i$. If we consider this as the row space of a matrix, it becomes clear that these operations are equivalent to elementary column operations. Thus they may change the space spanned by the row set but not its dimension, so by definition these operations result in a download-equivalent repair scheme.
\end{proof}

\subsection{Perfect Bandwidth Repair Schemes}\label{sec:optbw}

Recall from Section~\ref{sec:related} that an $(n,k)$ MDS array code with $t \leq n-k$ substripes has \emph{perfect bandwidth} if it achieves the trivial lower bound on bandwidth, $b = k+t-1$.

\begin{observation}\label{rem:d}
Any perfect bandwidth $(n,k)$ MDS code with $t$ linear substripes must have $d = k+t-1$ nodes in each repair set, and download a single symbol from each. This follows from the cut-set bound \cite{dimakis2010}
\[b \geq \frac{td}{d-k+1}\]
because $d \leq b$ trivially (each node in the repair set contributes at least one symbol), and we can only achieve $b \leq k+t-1$ if $d \geq b$.
\end{observation}

Per Theorem~\ref{thm:matrixchar}, a repair scheme for node $i^*$ from a set of other nodes $S$ can be characterized as a set of repair matrices $\{W^{(0)}, \dots, W^{(t-1)}\}$ in $\mathcal{C}^\perp$ with
\begin{align*}
\dim(\{\rowof{W}{j}{i^*} \mid j \in [0,t-1]\}) &=t \\
\sum_{i \in S} \dim(\{\rowof{W}{j}{i} \mid j \in [0,t-1]\}) &\leq b
\end{align*}
From Observation~\ref{rem:d}, we must have $|S| = k+t-1$ and $\dim(\{\rowof{W}{j}{i} \mid j \in [0,t-1]\}) = 1 \ \forall i \in S$, for any perfect bandwidth repair scheme. With this observation, as well as the preceding discussion in Section~\ref{sec:pbchar}, we note some properties of perfect bandwidth repair schemes for piggybacking codes which will be useful later.

\begin{definition}
We say a piggybacking repair scheme $\{W^{(0)}, \dots, W^{(t-1)}\}$ for node $i^*$ from a set $S$ is in \emph{standard form} if the following conditions hold; see Figure~\ref{fig:stdform} for an illustration.
\begin{enumerate}
\item We can partition $S$ such that for some $T \subset S$, $|T| = k-1$, we have $S = T \cup \{r_0, \dots, r_{t-1}\}$, where repair matrix $W^{(j)}$ is only non-zero on rows $T \cup \{r_j, i^*\}$.
\item The last column of each repair matrix has exactly $k+1$ non-zero entries, i.e., $\text{wt}(\colof{w}{j}{t-1}) = k+1 \ \forall j$.
\end{enumerate}
\end{definition}

\begin{figure}[h]
\begin{tikzpicture}[scale=.35]

\begin{scope}
\draw (0,-2) rectangle (6,10);
\node at (3,9) {1\ 0\ 0\ $\cdots$\ 1};
\draw[dashed] (.2,8) to (5.8,8);
\node at (3,7) {2\ 2\ 2\ $\cdots$\ 2};
\node at (3,6) {1\ 0\ -1\ $\cdots$\ 1};
\draw[dashed] (.2,5) to (5.8,5);
\node at (3,4) {1\ 2\ 0\ $\cdots$\ 1};
\node at (3,3) {0\ 0\ 0\ $\cdots$\ 0};
\node at (3,2) {0\ 0\ 0\ $\cdots$\ 0};
\draw[dashed] (.2,1) to (5.8,1);
\node at (3,0) {0\ 0\ 0\ $\cdots$\ 0};
\node at (3,-1) {0\ 0\ 0\ $\cdots$\ 0};
\node at (3,-3) {$W^{(0)}$};
\end{scope}

\begin{scope}[xshift=8cm]
\draw (0,-2) rectangle (6,10);
\node at (3,9) {0\ 1\ 0\ $\cdots$\ 1};
\draw[dashed] (.2,8) to (5.8,8);
\node at (3,7) {1\ 1\ 1\ $\cdots$\ 1};
\node at (3,6) {2\ 0\ -2\ $\cdots$\ 2};
\draw[dashed] (.2,5) to (5.8,5);
\node at (3,4) {0\ 0\ 0\ $\cdots$\ 0};
\node at (3,3) {1\ 1\ 2\ $\cdots$\ 1};
\node at (3,2) {0\ 0\ 0\ $\cdots$\ 0};
\draw[dashed] (.2,1) to (5.8,1);
\node at (3,0) {0\ 0\ 0\ $\cdots$\ 0};
\node at (3,-1) {0\ 0\ 0\ $\cdots$\ 0};
\node at (3,-3) {$W^{(1)}$};
\end{scope}

\node at (17,5) {$\cdots$};

\begin{scope}[xshift=20cm]
\draw (0,-2) rectangle (6,10);
\node at (3,9) {0\ 0\ 0\ $\cdots$\ 1};
\draw[dashed] (.2,8) to (5.8,8);
\node at (3,7) {2\ 2\ 2\ $\cdots$\ 2};
\node at (3,6) {2\ 0\ -2\ $\cdots$\ 2};
\draw[dashed] (.2,5) to (5.8,5);
\node at (3,4) {0\ 0\ 0\ $\cdots$\ 0};
\node at (3,3) {0\ 0\ 0\ $\cdots$\ 0};
\node at (3,2) {2\ 0\ 0\ $\cdots$\ 1};
\draw[dashed] (.2,1) to (5.8,1);
\node at (3,0) {0\ 0\ 0\ $\cdots$\ 0};
\node at (3,-1) {0\ 0\ 0\ $\cdots$\ 0};
\node at (3,-3) {$W^{(t-1)}$};
\end{scope}

\begin{scope}[xshift=27cm]
\node[anchor=west] at (0,9) {Row $i^*$ has full dimension $t$};
\node[anchor=west] at (0,6.5) {Rows $i \in T$ have low dimension};
\node[anchor=west] at (0,3) {Row $r_j$ is exclusive to $W^{(j)}$};
\node[anchor=west] at (0,-.5) {All other $n - k - t$ rows are $\bm{0}$};
\end{scope}

\begin{scope}[xshift=-1cm]
\node[anchor=east] at (0,9) {$i^*$};
\node[anchor=east] at (0,6.5) {$T$};
\node[anchor=east] at (0,4) {$r_0$}; 
\node[anchor=east] at (0,3) {$\cdots$};
\node[anchor=east] at (0,2) {$r_{t-1}$};
\end{scope}
\end{tikzpicture}
\caption{An example illustrating the structure of a linear repair scheme in standard form with $q = 3$.}\label{fig:stdform}
\end{figure}

\begin{observation}\label{rem:k+1}
Consider a perfect bandwidth $(n,k)$ piggybacking code $\mathcal{C}$ with $t$ substripes, base code $\mathcal{C}_0$ with generator matrix $F$, and piggybacking matrices $\{\Pij{i}{j} \mid i \in [0,t-2],j \in [i+1,t-1]\}$. If $\{W^{(0)}, \dots, W^{(t-1)}\}$ is any repair scheme for $\mathcal{C}$ and some node $i^*$, then the rightmost \emph{non-zero} column of each $W^{(j)}$ must have at least $k+1$ non-zeros. This is because $W^{(j)} \in \mathcal{C}^\perp$ requires
\[F \colof{W}{j}{i} + \Pij{i}{i+1} \colof{W}{j}{i+1} + \cdots + \Pij{i}{t-1} \colof{W}{j}{t-1} = \bm{0}^T \ \forall i\]
by Lemma~\ref{lem:cols}, so if $i$ is the rightmost non-zero column we get $F \colof{W}{j}{i} = \bm{0}^T$. But since the base code $\mathcal{C}_0$ is MDS, any $k$ columns of the generator matrix $F$ are linearly independent, so this equation can only hold if $\colof{W}{j}{i}$ has at least $k+1$ non-zeros. Note also that $W^{(j)}$ must have a rightmost non-zero column. Otherwise, $W^{(j)}$ would be entirely zero, including row $i^*$. But then $\dim \{\rowof{W}{j}{i^*} \mid j \in [0,t-1]\} \leq t-1$, which contradicts the requirements for a valid repair scheme from Theorem~\ref{thm:matrixchar}.
\end{observation}

\begin{lemma}\label{lem:std}
Let $\mathcal{C}$ be a perfect bandwidth $(n,k)$ piggybacking code with $t$ substripes. Then for any failed node $i^*$, repair set $S$, and corresponding repair scheme $\{W^{(0)}, \dots, W^{(t-1)}\}$, there exists an equivalent repair scheme in standard form.
\end{lemma}

\begin{proof}
From Observation~\ref{rem:d}, we know that the matrices $\{W^{(0)}, \dots, W^{(t-1)}\}$ are non-zero on exactly $k+t$ rows, and each row of $S$ has dimension 1. Furthermore, from Lemma~\ref{lem:equi}, we know that scaling and adding two repair matrices will result in an equivalent repair scheme. Thus to obtain an equivalent repair scheme in standard form, we proceed as follows.
\begin{enumerate}
\item Choose any $T \subset S$, $|T| = k-1$ to be the set of shared non-zero rows.
\item Go through the repair matrices in order. For each $W^{(j)}$,
	\begin{itemize}
	\item Pick $r_j$ to be any remaining non-zero row in $W^{(j)}$ not in $T$ or equal to $i^*$. Such a row must exist because otherwise $W^{(j)}$ could have at most $k$ non-zero rows, but this would contradict Observation~\ref{rem:k+1}.
	\item Add multiples of $W^{(j)}$ to each other repair matrix to zero out row $r_j$. This can be done because row $r_j \in S$ has dimension 1.
	\end{itemize}
\end{enumerate}
The above procedure results in each $W^{(j)}$ being non-zero only on the $k+1$ rows $T \cup \{r_j, i^*\}$ since every other row $r_i, i \neq j$ was zeroed out using matrix $i$; subsequent operations would not affect this row because every matrix added to $W^{(j)}$ had also had row $r_i$ zeroed out.

Furthermore, the last column of $W^{(j)}$ must have exactly $k+1$ non-zeros: By Observation~\ref{rem:k+1}, there are exactly $k+1$ non-zero rows in $W^{(j)}$, and either the last column is zero or has exactly $k+1$ non-zeros. But if $\colof{W}{j}{t-1}$ were zero, it would also be zero on the $k-1 \geq 1$ rows $T$, each of which has dimension 1 across the repair matrices. This would imply that every repair matrix has a last column of weight less than $k+1$ and therefore all zero by Observation~\ref{rem:k+1}. This would make it impossible to have row $i^*$ have full rank, so we have a contradiction, and the equivalent repair scheme obtained by the above procedure must be in standard form.
\end{proof}

\begin{lemma}\label{lem:q}
Let $\mathcal{C}$ be a perfect bandwidth $(n,k)$ piggybacking code over $\mathbb{F}_q$ with $t=2$ substripes. Then $q \geq k+1$.
\end{lemma}

\begin{proof}
Consider a repair scheme for a node $i^*$ from a set $S$ for such a code, $\{W^{(0)}, W^{(1)}\}$. Because we can perform linear operations on the repair matrices and get an equivalent repair scheme by Lemma~\ref{lem:equi}, we can assume without loss of generality that $\rowof{W}{0}{i^*} = [1 \ 0], \rowof{W}{1}{i^*} = [0 \ 1]$.

Recall that from Observation~\ref{rem:d}, the number of non-zero rows in the repair matrices is $k+t$ including $i^*$, so $|S| = k+t-1 = k+1$. Since the rows $S$ have dimension 1, and by Observation~\ref{rem:k+1} the rightmost non-zero column of each repair matrix has at least $k+1$ non-zeros, we observe that every row in $S$ must be non-zero in $W^{(0)}$. Since the rows of $S$ have dimension 1, this means we can express $\rowof{W}{1}{i} = \kappa_i \rowof{W}{0}{i}$ for every $i \in S$ and some $\kappa_i \in \mathbb{F}_q$.

If any $\kappa_i = \kappa_j$ for $i \neq j$, then consider $W^{(1)} - \kappa_i W^{(0)}$. From Lemma~\ref{lem:equi}, this would still be a dual codeword of $\mathcal{C}$, but clearly it has at least one non-zero in the last column (row $i^*$) and at most $k$ non-zeros in the last column (rows $i^*$ and $S \backslash \{i,j\}$). This contradicts Observation~\ref{rem:k+1}, implying that all the $\kappa_i$'s must be distinct. This necessitates $k+1$ distinct scalars in $\mathbb{F}_q$ so we conclude $q \geq k+1$ as desired.
\end{proof}


\section{``Any $d$'' Regime}\label{sec:any}
As noted in Observation~\ref{rem:d}, perfect bandwidth piggybacking codes with low substriping $t \leq n-k$ must have locality $d = k+t-1$ other nodes. In this section, we consider the regime where, when a node fails, \emph{any} set of $d$ other nodes must be able to repair it with bandwidth $b$ (as opposed to the less strict ``some $d$'' regime, which is treated in Section~\ref{sec:some}). For general linear erasure codes in this regime, \cite{shah2010} showed that for $t \leq k-3$, the cut-set bound is not achievable when only a single symbol is downloaded from each node of the repair set (and thus perfect bandwidth is not achievable). However, they also constructed a code which does achieve the cut-set bound for any $t \geq k - 2$, at least for repairing the systematic nodes. In this section, we will show that piggybacking codes cannot achieve perfect bandwidth for $k \geq 3$, and thus are strictly weaker than general linear codes in this regime.

\subsection{Non-achievability of Perfect Bandwidth for $k \geq 3$}
In this section, we prove the following theorem.
\begin{theorem}\label{thm:anyimposs}
No $(n,k)$ piggybacking code with $k \geq 3$ and $t$ substripes can achieve perfect bandwidth in the regime where any $d = k+t-1$ nodes must be able to repair a failed node.
\end{theorem}

To prove this, we will assume the existence of such a code, and consider an arbitrary pair of repair schemes for two different nodes with the same set of non-zero rows in their repair matrices. We assume these repair schemes are in standard form, which can be done without loss of generality, per Lemma~\ref{lem:std}. Finally, we show that the requirements for both schemes to have bandwidth $k+t-1$ render it impossible for the first scheme to recover its failed node. This will demonstrate that no piggybacking code for $k \geq 3$ can achieve perfect bandwidth, even for systematic nodes (in fact, even for only two nodes), and thus piggybacking codes are strictly less powerful than general linear codes in the ``any $d$'' regime.

\begin{lemma}\label{lem:schemes}
No $(n,k)$ piggybacking code with $k \geq 3$ and $t$ substripes can have two bandwidth $b=k+t-1$ repair schemes for two different failed nodes with the same set of $k+t$ non-zero rows.
\end{lemma}
\begin{proof}
Assume to obtain a contradiction that for some $k \geq 3$, some $(n,k)$ piggybacking code $\mathcal{C}$---over $\mathbb{F}_q$ with $t$ substripes, base code $\mathcal{C}_0$ generated by $F$, and piggybacking matrices $\{\Pij{i}{j} \mid i \in [0,t-2], j \in [i+1,t-1]\}$---obtains a repair bandwidth of $k+t-1$ for two such repair schemes. Let these schemes be $\{W^{(0)}, \dots, W^{(t-1)}\}$ which repairs node $i^*_W$ and $\{V^{(0)}, \dots, V^{(t-1)}\}$ which repairs $i^*_V \neq i^*_W$, both with bandwidth $k+t-1$. Furthermore, let the set of non-zero row indices be $S$, such that the $W$'s repair $i^*_W$ from $S_W = S \setminus \{i^*_W\}$ and the $V$'s repair $i^*_V$ from $S_V = S \setminus \{i^*_V\}$.

By Lemma~\ref{lem:std}, we can assume without loss of generality that these repair matrices are in standard form. Namely, each repair matrix has exactly $k+1$ non-zero rows; we can choose the shared sets of rows $T = T_W \cup \{i^*_W\} = T_V \cup \{i^*_V\}$ such that all the repair matrices share $|T| = k \geq 3$ non-zero rows including $i^*_W, i^*_V,$ and some other row $r_\text{shared}$; and we can renumber the matrices such that for all $j$, $W^{(j)}$ and $V^{(j)}$ share the same non-zero rows. Furthermore, for all $j$, $\colof{W}{j}{t-1}$ and $\colof{V}{j}{t-1}$ have the same $k+1$ non-zero rows, both live in $\mathcal{C}_0^\perp$, and are non-zero because the matrices are in standard form. Thus they live in the same one-dimensional subspace, so if we define $\bm{x^{(j)}} = \colof{W}{j}{t-1}$, we can scale $V^{(j)}$ such that $\colof{V}{j}{t-1} = \bm{x^{(j)}} = \colof{W}{j}{t-1}$; recall that scaling repair matrices still results in an equivalent repair scheme by Lemma~\ref{lem:equi}.

Now that we have fixed the last column of every matrix, consider the next-to-last columns. Define $\bm{y^{(j)}} = \colof{W}{j}{t-2}$. Now, because each $V^{(j)} \in \mathcal{C}^\perp$, we must have $F \colof{V}{j}{t-2} + \Pij{t-2}{t-1} \colof{V}{j}{t-1} = F \colof{V}{j}{t-2} + \Pij{t-2}{t-1} \bm{x^{(j)}} = \bm{0}^T$. Since we know $\colof{V}{j}{t-2}$ is only non-zero on some fixed $k+1$ rows, any solution to this equation (possible next-to-last column of $V^{(j)}$) can be represented as one fixed solution with only these $k+1$ rows non-zero, plus some vector in $\mathcal{C}_0^\perp$ with only these $k+1$ rows non-zero. Note, however, that $\bm{y^{(j)}}$ is one such fixed solution (since $W^{(j)} \in \mathcal{C}^\perp$ implies $F \colof{W}{j}{t-2} + \Pij{t-2}{t-1} \colof{W}{j}{t-1} = F \colof{W}{j}{t-2} + \Pij{t-2}{t-1} \bm{x^{(j)}} = \bm{0}^T$), and all vectors in $\mathcal{C}_0^\perp$ with only these $k+1$ rows non-zero are scalar multiples of $\bm{x}^{(j)}$. Thus for some $\lambda^{(j)} \in \mathbb{F}_q$ we must have $\colof{V}{j}{t-2} = \bm{y^{(j)}} + \lambda^{(j)} \bm{x^{(j)}}$. Thus we have
\[\forall j \qquad W^{(j)} = \begin{bmatrix}
		& \vertbar							& \vertbar	\\
\cdots	& \bm{y^{(j)}}							& \bm{x^{(j)}}	\\
		& \vertbar							& \vertbar	\\
\end{bmatrix}, \ V^{(j)} = \begin{bmatrix}
		& \vertbar							& \vertbar	\\
\cdots	& \bm{y^{(j)}} + \lambda^{(j)} \bm{x^{(j)}}	& \bm{x^{(j)}}	\\
		& \vertbar							& \vertbar	\\
\end{bmatrix}\]
Now, we recall the assumption that $\{W^{(0)}, \dots, W^{(t-1)}\}$ is a repair scheme for $i^*_W$ and $\{V^{(0)}, \dots, V^{(t-1)}\}$ for $i^*_V$, with bandwidth $k+t-1$, and all the repair matrices are non-zero on row $r_\text{shared} \neq i^*_W,i^*_V$. Thus by Observation~\ref{rem:d}, we know that $\dim ( \{\rowof{W}{j}{r_\text{shared}} \mid j \in [0,t-1] \} ) = \dim ( \{\rowof{V}{j}{r_\text{shared}} \mid j \in [0,t-1] \} ) = 1$. This must also hold when considering only the last two columns, so
\[\dim (\{ (\entryofvec{y}{j}{r_\text{shared}}, \ \entryofvec{x}{j}{r_\text{shared}} ) \ \mid \ j \in [0,t-1] \} ) = \dim ( \{ (\entryofvec{y}{j}{r_\text{shared}} + \lambda^{(j)}\entryofvec{x}{j}{r_\text{shared}}, \ \entryofvec{x}{j}{r_\text{shared}} ) \ \mid \ j \in [0,t-1] \} ) = 1\]
However, since row $r_\text{shared}$ is non-zero in every $W$ and $V$, we know $\entryofvec{x}{j}{r_\text{shared}}$ is non-zero for every $j$, and thus for both of these sets to have dimension 1, it must be that $\lambda^{(j)} = \lambda$ for every $j$.

But now consider row $i^*_W$. We know that in the $V$'s this row must have dimension 1 in order to achieve bandwidth $k+t-1$, which implies
\[\dim ( \{ (\entryofvec{y}{j}{i^*_W} + \lambda \entryofvec{x}{j}{i^*_W}, \ \entryofvec{x}{j}{i^*_W} ) \ \mid \ j \in [0,t-1] \} ) = 1\]
but then
\[\dim ( \{ (\entryofvec{y}{j}{i^*_W}, \ \entryofvec{x}{j}{i^*_W} ) \ \mid \ j \in [0,t-1] \} ) = 1\]
which implies that in the $W$'s, row $i^*_W$ has dimension at most $t-1$. But to be a valid scheme repairing node $i^*_W$, this row must have full dimension, giving us the desired contradiction.
\end{proof}

This lemma immediately leads to the desired result.
\begin{proof}[Proof of Theorem~\ref{thm:anyimposs}]
Assume to obtain a contradiction that a perfect bandwidth $(n,k)$ piggybacking code $\mathcal{C}$ with $k \geq 3$ and $t$ substripes does exist. Then there must exist repair schemes $\{W^{(0)}, \dots, W^{(t-1)}\}$ repairing node $i^*_W = 0$ from nodes $1, 2, \dots, k+t-1$, and $\{V^{(0)}, \dots, V^{(t-1)}\}$ repairing $i^*_V = 1$ using nodes $0, 2, 3, \dots, k+t-1$ each with bandwidth $k+t-1$. However, these repair schemes share the same set of non-zero rows, thus meeting the conditions of Lemma~\ref{lem:schemes}, so we have a contradiction.
\end{proof}

\subsection{Achievability for $k = 2$}\label{sec:anyk=2}
Although piggybacking codes cannot achieve perfect bandwidth for any $k \geq 3$, they can achieve it for $k=2$ and any $2 \leq t \leq n-k$, provided the field size is sufficiently large. In this section, we offer a non-constructive proof that such codes exist for any $2 \times n$ MDS base code generator matrix $F$ and $2 \leq t \leq n-k$, even if we only permit linebacking, i.e., piggybacking only in the last substripe (as in \cite{yang2015}). We proceed by considering the repair of any given node $i^*$ from a given set of $k+t-1 = t+1$ nodes, and counting the choices of piggybacking functions which fail to yield a valid repair scheme. Then we union bound over all choices of $i^*$ and sets of $t+1$ nodes to show that for sufficiently large field size $q$, some choice of piggybacking functions admits a repair scheme for every $i^*$ and set of $t+1$ nodes, with bandwidth $t+1$.

\begin{theorem}\label{thm:k=2}
Let $\cC_0 \subset \F_q^n$ be an MDS code with dimension $k=2$ and generator matrix $F \in \F_q^{k \times n}$, let $2 \leq t \leq n - k$, and suppose that $q$ is sufficiently large so that
$n \binom{n-1}{t+1} \left( 1 - \Pi_{i=1}^{t-1} (1 - \frac{1}{q^i})\right) < 1$.
Then there exists a linebacking code with base code $\cC_0$ and $t$ substripes, which achieves perfect bandwidth ($b = k+t-1 = t+1$), in the ``any $d$" regime.
\end{theorem}

\begin{proof}
Consider an arbitrary failed node $i^*$ and a repair set $S$ of size $t+1$.
Choose linebacking matrices $\{\Pj{i} \mid i \in [0,t-2]\}$ uniformly at random and let $\cC$ be the resulting linebacking code.  We will show that with positive probability, $\cC$ admits a valid repair scheme for $i^*$ from $S$.

We will construct a set of repair matrices.  Write $S = \{r^*, r_0, \ldots, r_{t-1}\}$.  We will choose matrices $W^{(0)}, \ldots, W^{(t-1)}$ with the following structure:
\begin{enumerate}
	\item $\rowof{W}{j}{i^*}$ may be nonzero for all $j$
	\item $\rowof{W}{j}{r_j}$ may be nonzero for all $j$
	\item $\entryof{W}{j}{r^*}{t-1}$ may be nonzero for all $j$
	\item all other entries of $W^{(j)}$ for all $j$ are zero.
\end{enumerate}
Now, we will fill in the nonzero entries of each $W^{(j)}$, so that $W^{(j)} \in \cC^\perp$.
First note that $\colof{W}{j}{t-1}$ is restricted to only have $3$ nonzero entries.  For each $j$, fix any choices for these nonzero entries so that $\colof{W}{j}{t-1} \neq 0$ and so that $F \colof{W}{j}{t-1} = 0$.  Notice that because $F \in \F_q^{2 \times n}$ and we have three nonzero entries to set, such a choice exists.  For $i \in [0,t-2]$, set
\[ \bm{z_i^{(j)}} = -P^{(i)} \colof{W}{j}{t-1}. \]
Thus, $\bm{z_i^{(j)}}$ is uniformly random (because $P^{(i)}$ is), and for a fixed $j$, the set $\{ \bm{z_i^{(j)}} \mid i \in [0,t-2] \}$ is independent.  Now for each $i<t-1$ and for each $j$, choose $\colof{W}{j}{i}$ so that it obeys the nonzero pattern above, and so that
\[ F \colof{W}{j}{i} = \bm{z_i^{(j)}}. \]
We can do this because $\colof{W}{j}{i}$ has two nonzero entries and each $2 \times 2$ minor of $F$ is full rank, since $\cC_0$ is MDS; moreover, the choice of $\colof{W}{j}{i}$ is thus also uniformly random, and for a fixed $j$, the set $\{ \colof{W}{j}{i} \mid i \in [0,t-2] \}$ is independent.  

By Corollary~\ref{cor:cols}, the matrices $W^{(j)}$, when chosen as above, are in $\cC^\perp$.  In order to constitute a set of optimal repair matrices for $\cC$, they thus need the additional properties that $\{ \rowof{W}{j}{r} \mid j \in [0,t-1] \}$ has dimension $1$ for all $r \in S$, and that $\{ \rowof{W}{j}{i^*} \mid j \in [0,t-1] \}$ has dimension $t$.  The first of these is satisfied by construction: for $r \in \{r_0,\ldots, r_{t-1}\}$, there is only one element of $\{ \rowof{W}{j}{r} \mid j \in [0,t-1] \}$ that is not zero; for $r = r^*$, every element of $\{ \rowof{W}{j}{r} \mid j \in [0,t-1] \}$ is a multiple of $\bm{e_{t-1}}$.   

Thus, it remains to compute the probability that the second of these occurs, namely, that
\[\dim(\{\rowof{W}{j}{i^*} \mid j \in [0,t-1]\}) = t.\] 
For a fixed $i^*$, these vectors are independent and uniformly random, so the probability that these are full rank is precisely
$\Pi_{i=1}^{t-1} \left(1 - \frac{1}{q^i}\right)$.  

Thus for any fixed $i^*$ and set $S$, the probability that a random choice of piggybacking matrices admits a valid repair scheme is $\Pi_{i=1}^{t-1} \left(1 - \frac{1}{q^i}\right)$. Since any set of $d = k+t-1 = t+1$ nodes must repair a failed node, and there are $n$ nodes, there are $n \binom{n-1}{t+1}$ choices for $i^*, S$. Thus union bounding over all such choices, there exist piggybacking matrices admitting a valid repair scheme for all $i^*,S$ provided $n \binom{n-1}{t+1} \left( 1 - \Pi_{i=1}^{t-1} \left(1 - \frac{1}{q^i}\right)\right) < 1$, which holds if $q$ is sufficiently large.
\end{proof}

\section{``Some $d$'' Regime}\label{sec:some}
In this section, we consider the regime where,
when a node fails, there must exist only some single set of $d = k+t-1$ nodes which can repair it with bandwidth $b$. As we shall see, this is a strictly weaker requirement than requiring \emph{any} set of $d$ nodes to repair the failed node, as we did in Section~\ref{sec:any}. It is also a popular regime for instantiating the piggybacking design framework, e.g., in \cite{kumar2015}, \cite{rashmi2013}, \cite{shangguan2016}, and \cite{yang2015}.

The ``some $d$" regime seems more difficult to get a handle on than the ``any $d$" regime, although some results do immediately transfer over.  For example, since perfect bandwidth linebacking codes for $k=2$ exist in the ``any $d$'' regime, then they exist in the ``some $d$'' regime as well.  Additionally, when $t = n-k$ and $d = k + t - 1 = n-1$, then the two regimes coincide, and so all of the results of Section~\ref{sec:any} still hold if $t = n-k$. This implies perfect bandwidth piggybacking codes for $t=n-k$ are still impossible in the ``some $d$'' regime, and since the constructions of \cite{shah2012} and \cite{suh2011} give perfect bandwidth MDS codes for $t=n-k$, this implies piggybacking codes are strictly weaker than general MDS codes in this regime as well.  However, the two regimes are not equivalent; in Section~\ref{sec:genk} below we exhibit an example of a piggybacking code achieving perfect bandwidth for $k=3$, which is impossible in the ``any $d$" regime.

While lower bounds are more difficult in the ``some $d$" regime, upper bounds (achievability results) are easier, and in Section~\ref{sec:somek=2} we strengthen our results for $k=2$ for the ``any $d$" regime.  More precisely, Theorem~\ref{thm:k=2} gives a non-constructive existence proof of perfect bandwidth piggybacking codes for $k=2$.  In Section~\ref{sec:somek=2} we strengthen this for $k=2,t=2$ by giving an explicit construction.

Finally, while we are unable to prove impossibility results for piggybacking codes in general in the ``some $d$" regime, we are able to prove impossibility results for linebacking codes.  In Section~\ref{sec:somelinebacking} we show that perfect bandwidth linebacking codes do not exist for $k \geq 3$ and for $t > \frac{ n-k + 1 }{2}$.

\subsection{General $k$}\label{sec:genk}
For $k \geq 3$, establishing impossibility results for perfect bandwidth piggybacking codes in the ``some $d$" regime seems quite difficult.  In particular, our approach of Theorem~\ref{thm:anyimposs} breaks down, because while Lemma~\ref{lem:schemes} still holds---no two schemes can have the same set of non-zero rows---in the ``some $d$" regime, it is easy to choose the repair sets for the failed nodes so that this does not occur.

Below, we show that there is a very good reason that we cannot match the strength of Theorem~\ref{thm:anyimposs} in the ``some $d$" regime, namely that it is not true!
We show by example (Theorem~\ref{thm:someconstruction}) that in fact there \em is \em a perfect bandwidth piggybacking code with $n=6,k=3,t=2,$ and $q=7$ in the ``some $d$" regime.  This establishes a separation between the ``some $d$" and ``any $d$" regimes, since Theorem~\ref{thm:anyimposs} shows that there is no such scheme in the ``any $d$" regime.

The counterexample in Theorem~\ref{thm:someconstruction} is for $t=2$, and it is natural to ask if impossibility results might still hold for larger $t$.  We also provide some results which shows that this may be the case.  
More precisely, we show in Theorem~\ref{thm:t-1} and the ensuing Corollary~\ref{cor:dect} that if there is a perfect bandwidth piggybacking code in the ``some $d$" regime with $t$ substripes, then there also exists one with $t-1$ substripes, down to $t=2$.  Thus, a negative result for any fixed $t_0$ would imply a negative result for all $t \geq t_0$, and a positive result for $t_0$ would imply a positive result for all $t \leq t_0$.

\begin{theorem}\label{thm:someconstruction}
There is a piggybacking code with $n=6,k=3,t=2,q=7$ which achieves perfect bandwidth in the ``some $d$" regime.
\end{theorem}
\begin{proof}
The example is given in Figure~\ref{fig:k=3}, along with the repair matrices.
\begin{figure}[h]
\begin{gather*}
F = \begin{bmatrix}
1 & 0 & 0 & 1 & 3 & 6 \\
0 & 1 & 0 & 4 & 6 & 6 \\
0 & 0 & 1 & 3 & 6 & 3
\end{bmatrix}, \qquad P = \begin{bmatrix}
0 & 0 & 0 & 1 & 0 & 0 \\
0 & 0 & 0 & 0 & 0 & 0 \\
0 & 0 & 0 & 0 & 1 & 0
\end{bmatrix} \\
i^* = 0: \ \begin{bmatrix}
4 & 1 \\
0 & 6 \\
0 & 0 \\
0 & 3 \\
0 & 0 \\
0 & 4
\end{bmatrix}, \qquad \begin{bmatrix}
2 & 1 \\
0 & 0 \\
0 & 2 \\
0 & 5 \\
0 & 0 \\
0 & 6
\end{bmatrix} \qquad
i^* = 1: \ \begin{bmatrix}
1 & 1 \\
6 & 5 \\
0 & 0 \\
0 & 2 \\
6 & 6 \\
0 & 0
\end{bmatrix}, \qquad \begin{bmatrix}
0 & 0 \\
6 & 1 \\
0 & 4 \\
0 & 3 \\
6 & 6 \\
0 & 0
\end{bmatrix} \qquad
i^* = 2: \ \begin{bmatrix}
0 & 1 \\
0 & 0 \\
2 & 6 \\
0 & 0 \\
0 & 5 \\
0 & 2
\end{bmatrix}, \qquad \begin{bmatrix}
0 & 0 \\
0 & 1 \\
5 & 5 \\
0 & 0 \\
0 & 2 \\
0 & 6
\end{bmatrix} \\
i^* = 3: \ \begin{bmatrix}
0 & 1 \\
0 & 0 \\
0 & 0 \\
6 & 4 \\
0 & 1 \\
3 & 1
\end{bmatrix}, \qquad \begin{bmatrix}
0 & 0 \\
3 & 1 \\
0 & 0 \\
4 & 1 \\
0 & 1 \\
5 & 4
\end{bmatrix} \qquad
i^* = 4: \ \begin{bmatrix}
0 & 0 \\
6 & 1 \\
0 & 0 \\
0 & 1 \\
3 & 1 \\
3 & 4
\end{bmatrix}, \qquad \begin{bmatrix}
0 & 0 \\
0 & 0 \\
0 & 1 \\
0 & 4 \\
6 & 3 \\
1 & 6
\end{bmatrix} \qquad
i^* = 5: \ \begin{bmatrix}
2 & 1 \\
0 & 0 \\
0 & 0 \\
0 & 4 \\
2 & 1 \\
5 & 1
\end{bmatrix}, \qquad \begin{bmatrix}
0 & 0 \\
0 & 0 \\
0 & 1 \\
0 & 4 \\
6 & 3 \\
1 & 6
\end{bmatrix}
\end{gather*}
\caption{A $(6,3)$ piggybacking code and repair scheme for $t=2,q=7$ in the ``some $d$'' regime.}\label{fig:k=3}
\end{figure}
\end{proof}

Finally, we show that decreasing $t$ does not make the problem of obtaining perfect bandwidth piggybacking codes any more difficult.
\begin{theorem}\label{thm:t-1}
Let $\mathcal{C}$ be an $(n,k)$ piggybacking code over $\mathbb{F}_q$ with $t$ substripes, base code generator matrix $F$, and piggybacking matrices $\{\Pij{i}{j} \mid i \in [0,t-2], j \in [i+1,t-1]\}$. If there exists a bandwidth $b$ repair scheme for $\mathcal{C}$, then we can construct a new $(n,k)$ piggybacking code $\mathcal{C}'$ over $\mathbb{F}_q$ with $t-1$ substripes and bandwidth at most $b-1$ in the regime where only \emph{some} $d = k+t-1$ nodes must be able to repair a failed node.
\end{theorem}
Before we prove Theorem~\ref{thm:t-1}, we note that this immediately implies that if there is a perfect bandwidth piggybacking code with $t$ substripes, then there is also one with $t-1$ substripes.
\begin{corollary}\label{cor:dect}
Let $\mathcal{C}$ be a perfect bandwidth $(n,k)$ piggybacking code over $\mathbb{F}_q$ with $t$ substripes. Then there exist piggybacking codes achieving perfect bandwidth for the same $n,k,q$ and any number of substripes up to $t$.
\end{corollary}
Additionally, Corollary~\ref{cor:dect}, along with Lemma~\ref{lem:q} about the required alphabet size, imply that perfect bandwidth piggybacking codes must have large alphabets:
\begin{corollary}\label{cor:k+1}
Let $\mathcal{C}$ be a perfect bandwidth $(n,k)$ piggybacking code over $\mathbb{F}_q$ in the regime where any $d = k+t-1$ nodes must be able to repair a failed node. Then $q \geq k+1$.
\end{corollary}
\begin{proof}
Per Lemma~\ref{lem:q}, any piggybacking code achieving perfect bandwidth for $t=2$ must have $q \geq k+1$, and by Corollary~\ref{cor:dect}, if a piggybacking code exists for any $t$ one exists for $t=2$ with the same $n,k,q$.
\end{proof}

Finally, we prove Theorem~\ref{thm:t-1}.
\begin{proof}[Proof of Theorem~\ref{thm:t-1}]
We define the new code $\mathcal{C}'$ with the same parameters $n,k,q$ and base code as $\mathcal{C}$, but with $t-1$ substripes, such that the new piggybacking matrices are $P^{(i,j)'} = \Pij{i+1}{j+1}$. This essentially amounts to removing the $0^{th}$ substripe of both the message and each node, along with the corresponding piggybacking functions.

Now we want to show that a valid repair scheme exists for $\mathcal{C}'$ with bandwidth at most $b-1$. Consider the repair scheme of $\mathcal{C}$ for node $i^*$, $\{W^{(0)}, \dots, W^{(t-1)}\}$. Since linear operations on this set of repair matrices result in an equivalent repair scheme, per Lemma~\ref{lem:equi}, we can assume without loss of generality that $\rowof{W}{0}{i^*} = [1 \ 0 \cdots 0]$. 
Now, there must exist some entry of $W^{(0)}$, $\entryof{W}{0}{i_{\text{fixed}}}{j_\text{fixed}} \neq 0$ such that $i_\text{fixed} \neq i^*$; otherwise only row $i^*$ would be non-zero, which would contradict Observation~\ref{rem:k+1}. Thus we can use this entry to zero out the $i_\text{fixed},j_\text{fixed}^{th}$ entry of every other repair matrix, by adding multiples of $W^{(0)}$.

Now we obtain a repair scheme of $\mathcal{C}'$ for $i^*$ from $\{W^{(0)}, \dots, W^{(t-1)}\}$ by deleting the $0^{th}$ column of each repair matrix and removing $W^{(0)}$ to obtain $\{W^{(1)'}, \dots, W^{(t-1)'}\}$. Again, this essentially amounts to deleting the contribution of the $0^{th}$ substripe of each node to the repair process, and deleting the repair matrix which solely recovers the $0^{th}$ substripe. Note that each new repair matrix $W^{(j)'}$ is in $\mathcal{C}'^\perp$: Per Lemma~\ref{lem:cols}, the original matrix $W^{(j)} \in \mathcal{C}^\perp$ implies
\[ F \colof{W}{j}{i} + \Pij{i}{i+1} \colof{W}{j}{i+1} + \cdots + \Pij{i}{t-1} \colof{W}{j}{t-1} = \bm{0}^T \ \forall i \]
Thus for $W^{(j)'}$ we have
\[F w^{(j)'}_{\bullet, i} + P^{(i,i+1)'} w^{(j)'}_{\bullet, i+1} + \cdots + P^{(i,t-2)'} w^{(j)'}_{\bullet, t-2} = F w^{(j)}_{\bullet, i+1} + P^{(i+1,i+2)} w^{(j)}_{\bullet, i+2} + \cdots + P^{(i+1,t-1)} w^{(j)}_{\bullet, t-1} = \bm{0}^T \ \forall i\]
so $W^{(j)'} \in \mathcal{C}'^\perp$, as desired.

Furthermore, $\dim(\{w^{(j)'}_{i^*,\bullet} \mid j \in [1,t-1]\}) = t-1$. We know $\dim(\{\rowof{W}{j}{i^*} \mid j \in [0,t-1]\}) = t$, and the deleted row $\rowof{W}{0}{i^*}$ was $[1 \ 0 \cdots 0]$, so it contributes nothing to the dimension of the row set obtained by deleting the $0^{th}$ entry of each row, which must have dimension $t-1$ for the original space to have full dimension.

Finally, the bandwidth of the new repair scheme $\{W^{(1)'}, \dots, W^{(t-1)'}\}$ is at most $b-1$. Clearly, we have $\dim ( \{ w^{(j)'}_i \mid j \in [1,t-1]\} )  \leq \dim ( \{ w^{(j)}_i \mid j \in [0,t-1]\} ) \ \forall i$, since deleting one row and one coordinate cannot increase the dimension of the row set. Additionally, the dimension of row $i_\text{fixed}$ must have decreased by at least one, because the $i_\text{fixed},j_\text{fixed}^{th}$ entry was zeroed out in every repair matrix but $W^{(0)}$. Thus deleting row $w^{(0)}_{i_\text{fixed}}$ from the row set $\{ w^{(j)}_{i_\text{fixed}} \mid j \in [0,t-1]\}$ decreases the dimension. Since the dimension of at least one row set decreased, and no dimension increased, the total bandwidth decreased by at least one.

Thus for every $i^*$, we can obtain a valid repair scheme of $\mathcal{C}'$ with bandwidth $b-1$ from the repair scheme of $\mathcal{C}$. Therefore given a piggybacking code $\mathcal{C}$ with $t$ substripes and bandwidth $b$, we can construct a piggybacking code with the same $n,k,q$, and $F$; $t-1$ substripes; and bandwidth at most $b-1$.
\end{proof}

\subsection{$k = 2$}\label{sec:somek=2}
Theorem~\ref{thm:k=2}, which guarantees the existence of optimal repair schemes for $k=2$ in the ``any $d$" regime, also applies in the ``some $d$" regime.\footnote{
In fact, since only \emph{some} set of $d$ nodes must be able to repair a given failed node, we need only union bound over choices of $i^*$ (not sets of $d$ nodes), and thus $q$ must only satisfy $n \left(1 - \Pi_{i=1}^{t-1} \left(1 - \frac{1}{q^i}\right)\right) < 1$.
}
  However, the requirement that only some repair set of size $d$ exist is weaker, and we can improve on Theorem~\ref{thm:k=2} in the ``some $d$" regime, giving an explicit construction of an optimal piggybacking code for $k=2,t=2$.

\begin{theorem}\label{thm:somek=2}
Let $\cC_0 \subset \F_q^n$ be an MDS code with dimension $k=2$ and generator matrix $F \in \F_q^{k \times n}$, let $t=2$, and suppose that $n \geq 4$, $q \geq 3$.  
Then there is an explicit construction of a piggybacking code with base code $\cC_0$ and with $t$ substripes, which achieves perfect bandwidth $b = k + t - 1 = 3$, in the ``some $d$" regime.
\end{theorem}
\begin{proof}
For $t=2$, there is only one piggybacking matrix, $\Pij{0}{1} = P$. We choose $P$ to be all zero, except for some column $i$ such that $\entryofF{1}{i} \neq 0$ (which much exist for $\cC_0$ to be MDS and $F$ therefore rank $k=2$); let $p_{0,i} = 1$.

First, we exhibit the repair scheme for any $i^* \neq i$. Pick the repair set $S = \{i, i_0, i_1\}$ for any rows $i_0, i_1 \neq i, i^*$. The two repair matrices will be
\[W^{(0)} = \begin{blockarray}{ccc}
\begin{block}{c[cc]}
i^*	& 1	& 0	\\
i & -\entryofF{1}{i^*} / \entryofF{1}{i} & -\entryof{W}{0}{i}{0} \entryofF{0}{i} - \entryofF{0}{i^*} \\
i_0 & 0 & \\
i_1 & 0 & \\
\end{block}
\end{blockarray}, \qquad
W^{(1)} = \begin{blockarray}{ccc}
\begin{block}{c[cc]}
i^*	& 0	& 1	\\
i & 0 & 0 \\
i_0 & 0 & \\
i_1 & 0 & \\
\end{block}
\end{blockarray}
\]
where we omit the zero rows and the missing entries are fully determined by the requirement that the last columns are in $\mathcal{C}_0^\perp$. Clearly this satisfies row $i^*$ having full dimension and has bandwidth $k+t-1 = 3$. We allow the missing entries to be chosen such that the last columns are in $\mathcal{C}_0^\perp$, and note that
\[F \colof{W}{0}{0} + P \colof{W}{0}{1} = \left[\entryofF{0}{i^*} - \entryofF{0}{i} \frac{\entryofF{1}{i^*}}{\entryofF{1}{i}}, \ \entryofF{1}{i^*} - \entryofF{1}{i} \frac{\entryofF{1}{i^*}}{\entryofF{1}{i}}\right]^T + \left[\frac{\entryofF{1}{i^*}}{\entryofF{1}{i}} \entryofF{0}{i} - \entryofF{0}{i^*}, \ 0\right]^T = \bm{0}^T\]
and
\[F \colof{W}{1}{0} + P \colof{W}{1}{1} = [0,0]^T + [0, 0]^T = \bm{0}^T\]
so the $W$'s are in $\mathcal{C}^\perp$. Thus we have exhibited a valid repair scheme for any $i^* \neq i$.

When $i^* = i$, the argument is slightly more complicated. We pick any repair set $S = \{i_0, i_1, i_2\}$ not including $i^*$. The repair matrices will be
\[W^{(0)} = \begin{blockarray}{ccc}
\begin{block}{c[cc]}
i^*=i	& 1	& 0	\\
i_0 & \entryof{W}{0}{i_0}{0} & \entryof{W}{0}{i_0}{1} \\
i_1 & \entryof{W}{0}{i_1}{0} & \\
i_2 & 0 & \\
\end{block}
\end{blockarray}, \qquad
W^{(1)} = \begin{blockarray}{ccc}
\begin{block}{c[cc]}
i^*=i & 0 & 1	\\
i_0 & \kappa_0 \entryof{W}{0}{i_0}{0} & \kappa_0 \entryof{W}{0}{i_0}{1} \\
i_1 & \kappa_1 \entryof{W}{0}{i_1}{0} & \\
i_2 & 0 & \\
\end{block}
\end{blockarray}
\]
where $\kappa_0, \kappa_1 \in \mathbb{F}_q$, we omit the zero rows, and the non-constant entries are yet to be determined. To have $W^{(0)} \in \mathcal{C}^\perp$, we need
\begin{equation}\label{eq:W0}
F \colof{W}{0}{0} + P \colof{W}{0}{1} = \begin{bmatrix}
\entryofF{0}{i} & \entryofF{0}{i_0} & \entryofF{0}{i_1} \\
\entryofF{1}{i} & \entryofF{1}{i_0} & \entryofF{1}{i_1}
\end{bmatrix} [1, \entryof{W}{0}{i_0}{0}, \entryof{W}{0}{i_1}{0}]^T = \bm{0}^T
\end{equation}
Since $\cC_0$ is MDS, this uniquely determines $\entryof{W}{0}{i_0}{0}, \entryof{W}{0}{i_1}{0} \neq 0$. Furthermore, to have $W^{(1)} \in \mathcal{C}^\perp$, we need
\begin{equation}\label{eq:W1}
F \colof{W}{1}{0} + P \colof{W}{1}{1} = \begin{bmatrix}
\entryofF{0}{i_0} & \entryofF{0}{i_1} \\
\entryofF{1}{i_0} & \entryofF{1}{i_1}
\end{bmatrix} [\kappa_0 \entryof{W}{0}{i_0}{0}, \kappa_1 \entryof{W}{0}{i_1}{0}]^T + [1,0]^T = \bm{0}^T
\end{equation}
Again, since $\cC_0$ is MDS and we know $\entryof{W}{0}{i_0}{0}, \entryof{W}{0}{i_1}{0} \neq 0$, this uniquely determines $\kappa_0, \kappa_1$. Note also that $\kappa_0 \neq \kappa_1$, since if they were equal that would imply either $\entryofF{1}{i_0} \entryof{W}{0}{i_0}{0} + \entryofF{1}{i_1} \entryof{W}{0}{i_1}{0} = 0$ or $\kappa_0 = \kappa_1 = 0$. The former would violate equation~\ref{eq:W0} since we specifically chose $\entryofF{1}{i} \neq 0$, and the latter clearly violates equation~\ref{eq:W1}.

Our goal now is to show that we can choose $W^{(0)}, W^{(1)}$ to simultaneously satisfy the known values above while also having their last columns in $\mathcal{C}_0^\perp$; this will ensure the repair scheme has row $i^* = i$ full rank, bandwidth $k+t-1=3$, and repair matrices in $\mathcal{C}^\perp$. We will show this by considering the $q-1$ possible non-zero choices for $\entryof{W}{0}{i_0}{1}$. We fix the $0^{th}$ columns and $i_0^{th}$ rows (since we know their desired values), and determine the last two unknowns in each last column such that the last columns are in $\mathcal{C}_0^\perp$. For some choice of $\entryof{W}{0}{i_0}{1}$, this must give the desired value of $\entryof{W}{1}{i_1}{1} / \entryof{W}{0}{i_1}{1}$, which needs to equal $\kappa_1$ for the repair scheme to have the desired bandwidth. This is because each non-zero choice of $\entryof{W}{0}{i_0}{1}$ must give a distinct value not equal to $\kappa_0$ for $\entryof{W}{1}{i_1}{1} / \entryof{W}{0}{i_1}{1}$. (Note that this is never division by zero because the last columns are in $\mathcal{C}_0^\perp$ and thus have a minimum weight of $k+1 = 3$, so if $\entryof{W}{0}{i_0}{1} \neq 0$, $\entryof{W}{0}{i_1}{1} \neq 0$ as well.) The resulting values of $\entryof{W}{1}{i_1}{1} / \entryof{W}{0}{i_1}{1}$ never equal $\kappa_0$, because if one did, $\colof{W}{1}{1} - \kappa_0 \colof{W}{0}{1}$ would still be in $\mathcal{C}_0^\perp$ but would have weight 1 or 2, which is less than the minimum weight. Also, if any two resulting values of $\entryof{W}{1}{i_1}{1} / \entryof{W}{0}{i_1}{1}$ were both equal to some $\kappa$, the difference between the two resulting values of $\colof{W}{1}{1} - \kappa \colof{W}{0}{1}$ for the two choices of $\entryof{W}{0}{i_0}{1}$ would be zero in rows $i, i_1$ but not in $i_0$, so it again violates the minimum weight for vectors in $\mathcal{C}_0^\perp$. Thus each possible value (not including $\kappa_0$) for $\entryof{W}{1}{i_1}{1} / \entryof{W}{0}{i_1}{1}$ will occur exactly once over choices of $\entryof{W}{0}{i_0}{1}$, so some such choice will give the desired value $\kappa_1$ and yield a valid repair scheme.
\end{proof}

\subsection{Linebacking codes}\label{sec:somelinebacking}
While we do not have general impossibility results for piggybacking codes in the ``some $d$" regime, we can prove impossibility results for \em linebacking \em codes in this regime.  We note that multiple constructions of piggybacking codes in the literature are in fact linebacking codes---including the design of \cite{yang2015} and any piggybacking codes with $t=2$ substripes, such as two constructions of \cite{rashmi2013}---and so such lower bounds provide useful design insights.

Our main theorem in this section is the following.
\begin{theorem}\label{thm:linebacking}
No $(n, k)$ linebacking code with $k\geq 3$ and $t > \frac{n-k+1}{2}$ substripes can achieve perfect bandwidth.
\end{theorem}
We remark that the constraint on $t$ is tight in the sense that Theorem~\ref{thm:someconstruction} gives an example of a perfect bandwidth linebacking code with $t = \frac{n-k+1}{2}$.  

In addition, we prove a stronger impossibility result for the special case of linebacking codes where every repair scheme (for systematic nodes) uses all remaining systematic nodes in its repair set. 
This impossibility result is notable because many piggybacking and linebacking constructions do use all the remaining systematic nodes  when repairing a failed node \cite{rashmi2013} \cite{yang2015}, and this result suggests that such an approach may limit the achievable repair bandwidth.
Our main result here is the following, which is a corollary of Theorem~\ref{thm:linesyst} that we will state and prove below.
\begin{restatable}{corollary}{cor}
\label{cor:linesyst}
No $(n,k)$ linebacking code with $k \geq 3$, $t > \frac{n-k-1}{\sqrt{k}}$ substripes, and which uses all remaining systematic nodes to repair a failed node can achieve perfect bandwidth.
\end{restatable}

\begin{remark}[Separation between linebacking and piggybacking?]
Note that in the regime where \emph{any} $d=k+t-1$ nodes must be able to repair a failed node, linebacking and general piggybacking were essentially equivalent in terms of ability to achieve perfect bandwidth: For $k\geq 3$ it was impossible for any piggybacking or linebacking code to achieve it, and for $k=2$ linebacking was sufficient provided the field size $q$ was large enough. 
However, in the ``some $d$" regime, it is possible that piggybacking codes are strictly more powerful than linebacking codes, in the sense that there may exists parameter regimes where linebacking codes cannot achieve perfect bandwidth but general piggybacking codes can.
We conjecture that this is the case, and note that an example of a perfect bandwidth piggybacking code in the regime where Theorem~\ref{thm:linebacking} holds would establish this.
\end{remark}

We begin by proving a lemma that will lead to the result of Theorem~\ref{thm:linebacking}.
\begin{lemma}\label{lem:lineschemes}
No $(n,k)$ linebacking code with $t$ substripes and $k \geq 3$ can have two bandwidth $b = k+t-1$ repair schemes for two different failed nodes, where each node participates in the other's repair \emph{and} the respective repair sets overlap by at least $k$ nodes.
\end{lemma}
\begin{proof}
We proceed along the lines of Lemma~\ref{lem:schemes}, but using only two repair matrices from each of the two repair schemes. Assume to obtain a contradiction that for some $k \geq 3$, some $(n,k)$ linebacking code $\mathcal{C}$---over $\mathbb{F}_q$ with $t$ substripes, base code $\mathcal{C}_0$ with generator matrix $F$, and piggybacking matrices $\{\Pj{i} \mid i \in [0,t-2]\}$---obtains a repair bandwidth of $k+t-1$ for two such repair schemes. Let these schemes be $\{W^{(0)}, \dots, W^{(t-1)}\}$ which repairs node $i^*_W$ and $\{V^{(0)}, \dots, V^{(t-1)}\}$ which repairs node $i^*_V \neq i^*_W$. By assumption, these two repair schemes share $k+2$ non-zero rows, including $i^*_W, i^*_V$.

By Lemma~\ref{lem:std}, we can assume without loss of generality that the repair matrices are in standard form. Namely, each repair matrix has exactly $k+1$ non-zero rows, and all the repair matrices share $k \geq 3$ rows including $i^*_W, i^*_V$, and some other row $r_\text{shared}$. Since the $W$'s and $V$'s share two additional rows besides these $k$ in common, we can renumber the matrices such that the pairs $W^{(0)}, V^{(0)}$ and $W^{(1)}, V^{(1)}$ each have the same $k+1$ non-zero rows. Furthermore, the last columns of $W^{(0)}, V^{(0)}$ (and $W^{(1)}, V^{(1)}$) live in $\mathcal{C}_0^\perp$ and have the same $k+1$ non-zero rows, so they live in a one-dimensional subspace. Thus, by Lemma~\ref{lem:equi}, we can scale $V^{(0)}$ (and $V^{(1)}$) such that $\colof{W}{0}{t-1} = \colof{V}{0}{t-1}$ (and $\colof{W}{1}{t-1} = \colof{V}{1}{t-1}$). Additionally, we know from Lemma~\ref{lem:downloadequi} that we can add the last columns of the $W$'s (or $V$'s) onto any previous column, provided we perform the same operation on every matrix, and obtain a download-equivalent repair scheme. Thus, for example, we can modify the repair matrices such that row $r_\text{shared}$ is zeroed out in $W^{(0)}$ and $V^{(0)}$ except in the last position. Now, observe that the last columns of $W^{(0)}$ and $V^{(0)}$ have already been made equal, and that each other column $j$ is constrained by $F \colof{W}{0}{j} + \Pj{j} \colof{W}{0}{t-1} = \bm{0}^T$ (and likewise for $V^{(0)}$) by Corollary~\ref{cor:cols}. Since each such constraint consists of $k$ linearly independent equations (since $F$ is MDS) and $k$ unknowns (since $k+1$ rows are non-zero and we fixed row $r_\text{shared}$), this constraint admits exactly one solution. But that implies each column of $W^{(0)}$ equals the corresponding column of $V^{(0)}$, and thus $W^{(0)} = V^{(0)}$.

However, recall that row $r_\text{shared}$ must have dimension 1 in both the $W$'s and $V$'s by Theorem~\ref{thm:matrixchar}. Since we already have $\colof{W}{1}{t-1} = \colof{V}{1}{t-1}$ which implies $\entryof{W}{1}{r_\text{shared}}{t-1} = \entryof{V}{1}{r_\text{shared}}{t-1}$, and $W^{(0)} = V^{(0)}$ which implies $\rowof{W}{0}{r_\text{shared}} = \rowof{V}{0}{r_\text{shared}}$, for row $r_\text{shared}$ to have dimension 1 in both repair schemes, we must have $\rowof{W}{1}{r_\text{shared}} = \rowof{V}{1}{r_\text{shared}}$. But again, this means the last columns are equal and each previous column has $k$ unknowns (with the same known values between $W^{(1)}$ and $V^{(1)}$), so we must have $W^{(1)} = V^{(1)}$. This gives us the desired contradiction, because in the $W$'s, row $i^*_W$ must have full rank, whereas in the $V$'s it must have dimension 1; these cannot both be true if $W^{(0)} = V^{(0)}$ and $W^{(1)} = V^{(1)}$.
\end{proof}

We are now ready to prove the main result of this section.
\begin{proof}[Proof of Theorem~\ref{thm:linebacking}]
Assume to obtain a contradiction that a perfect bandwidth $(n, k)$ linebacking code $\mathcal{C}$ with $t < \frac{n-k+1}{2}$ substripes does exist for some $k \geq 3$. Then there would necessarily be two repair schemes meeting the assumptions of Lemma~\ref{lem:lineschemes}. First, observe that there must exist a pair of nodes each of which participates in repairing the other. Consider a directed graph where edges go from each node to the nodes it repairs. Each node has $k+t-1$ in-edges so there are $n(k+t-1)$ edges. However, $n(k+t-1) > n(k+\frac{n-k+1}{2}-1) = \frac{n(n+k-1)}{2} > \binom{n}{2}$. But $\binom{n}{2}$ is the maximum number of edges a directed graph (with no self-loops) can have without having a 2-cycle; thus this graph has a 2-cycle meaning two nodes participate in each others' repair. Furthermore, these same two nodes must have an overlap in their repair sets of size at least $k$: In addition to repairing each other, they each have $k+t-2 > \frac{n+k-3}{2}$ helper nodes drawn from the remaining $n-2$ nodes, so the overlap is at least $2 (\frac{n+k-2}{2}) - (n - 2) = k$. Thus if such a linebacking code existed, it would necessarily have two repair schemes meeting the assumptions of Lemma~\ref{lem:lineschemes}, which is impossible.
\end{proof}

Note that for $t \leq \frac{n-k+1}{2}$, it is trivial to construct the repair sets (disregarding whether they admit valid repair schemes) to avoid any pair satisfying the assumptions of Lemma~\ref{lem:lineschemes}; for instance, each node can be repaired by the $k+t-1$ nodes immediately following it (mod $n$).

However, most existing piggybacking codes do not choose their repair sets this way. Most, including those of \cite{rashmi2013} and \cite{yang2015} (who use linebacking codes), always use all remaining systematic nodes in the repair. For linebacking codes, this further restricts the parameters for which perfect bandwidth may be achievable.

\begin{theorem}\label{thm:linesyst}
No $(n,k)$ linebacking code with $k \geq 3$, $t$ substripes where $t(t-1) > \frac{(n-k)(n-k-1)}{k}$, and which uses all remaining systematic nodes to repair a failed node can achieve perfect bandwidth.
\end{theorem}
\begin{proof}
Assume to obtain a contradiction that a perfect bandwidth $(n, k)$ linebacking code $\mathcal{C}$ with $k \geq 3$ and $t$ substripes where $t(t-1) > \frac{(n-k)(n-k-1)}{k}$ always uses all remaining systematic nodes to repair a failed node.

Consider only the repair of the systematic nodes. By assumption, for any pair of systematic nodes, each participates in the other's repair. Their repair sets overlap by at least $k$ if and only if there are at least 2 parity nodes which repair both, and each systematic node has $t$ parity nodes repairing it. Per \textcite{erdos1963}, the maximum number of sets of $t$ parity nodes such that no two sets have 2 parity nodes in common is $\frac{\binom{n-k}{2}}{\binom{t}{2}}$. Thus if $k > \frac{(n-k)(n-k-1)}{t(t-1)}$, two systematic nodes must share 2 helper parity nodes, and thus have an overlap of size at least $k$ in their repair sets as well as each participating in the other's repair. However, this meets the assumptions of Lemma~\ref{lem:lineschemes}, which is impossible.
\end{proof}

Simplifying the statement of Theorem~\ref{thm:linesyst} results in the Corollary~\ref{cor:linesyst} which we presented earlier.  We restate it here:

\cor*

This may suggest that linebacking codes can achieve better bandwidth if they do \emph{not} follow the standard practice of using all remaining systematic nodes in every repair, since the bound $t \leq \frac{n-k-1}{\sqrt{k}}$ is more restrictive than $t \leq \frac{n-k+1}{2}$ from Theorem~\ref{thm:linebacking} as $k$ grows.


\section{Conclusion}\label{sec:concl}
We adapted the framework of \cite{guruswami2016} in order to analyze the achievable bandwidth of piggybacking codes introduced by \textcite{rashmi2013} with scalar MDS base codes for low substriping $t \leq n-k$. In the regime where any $d$ nodes must be able to repair a failed node, we showed that for $k \geq 3$ piggybacking codes cannot achieve the lower bound on bandwidth, and thus are less powerful than general linear codes. We established by counterexample that this result does not extend to the regime where only some $d$ nodes repair a failed node (though piggybacking codes are still less powerful than general linear codes), and partially addressed the question of whether piggybacking codes can achieve the lower bound on bandwidth in this regime. We additionally gave impossibility results for linebacking, a subcategory of piggybacking in the style of \cite{yang2015}.

Some questions about the theoretical capabilities and limitations of piggybacking codes remain to be addressed, and we conclude with these.
\begin{enumerate}
\item When do there exist perfect bandwidth piggybacking codes for the ``some $d$'' regime and $k \geq 3$? When they do not exist, how close can piggybacking codes get to the lower bound on bandwidth?
\item Is linebacking less powerful than piggybacking?
\item Adding the (commonly used) assumption that all remaining systematic nodes assist in the repair of a failed node gave us a stronger impossibility result for linebacking in Corollary~\ref{cor:linesyst}. Does this assumption actually worsen the achievable bandwidth for piggybacking (or general) codes?
\item Our analysis of piggybacking codes was limited compared to the proposal of \cite{rashmi2013} in a few ways. How does the analysis change if we permit non-linear piggybacking functions, or allow a vector (rather than scalar) base code? What can we say about how piggybacking codes perform on other metrics such as data-read and computation as well as bandwidth?
\end{enumerate}


\section*{Acknowledgements}
We thank  Rashmi Vinayak for introducing the problem to us, and for very helpful correspondences.


\bibliographystyle{plain}
\bibliography{references}


\appendix

\section{Proof of Theorem~\ref{thm:matrixchar}}\label{app:thm2}
We follow the approach of \cite{guruswami2016}, which gives a similar result for scalar MDS codes.
We first show that 2 implies 1, that is, if the required repair matrices exist, then there is a linear repair scheme for node $i^*$ from $S$ with bandwidth $b$.

Suppose that $\{W^{(0)}, W^{(1)}, \ldots, W^{(t-1)} \} \subset \cC^\perp$ so that the only non-zero rows of $W^{(j)}$ are $i^*$ and $S$, and suppose that
\[ \mathrm{dim}( \{ \rowof{W}{j}{i^*} \mid j \in [0,t-1]\}) = t \]
and
\[ \sum_{i \neq i^*} \mathrm{dim}( \{ \rowof{W}{j}{i} \mid j \in [0,t-1]\}) \leq b \]
Define a repair scheme for a codeword $C \in \cC$ as follows.  
Let $b_i$ be the dimension of $\{ \rowof{W}{j}{i} \mid j \in [0,t-1]\}$, and let
Let $\bm{v}_0, \ldots, \bm{v}_{b_i-1}$ be a basis for the span of that space.  
Then the $b_i$ symbols of $\F_q$ returned by node $i$ are
\begin{equation}\label{eq:gotthis}
 \{ c_{i,\bullet} \cdot \bm{v}_\ell \mid \ell \in [0,b_i-1]\}
\end{equation}
We need to establish that, from these, we may recover $c_{i^*,\bullet}$, the contents of node $i^*$.
Using the fact that $\langle W^{(j)}, C \rangle = 0$ for all $C \in \cC$, we have
\[ c_{i^*, \bullet} \cdot \rowof{W}{j}{i^*} = - \sum_{i \in S} c_{i, \bullet} \cdot \rowof{W}{j}{i} \]
Thus, using the information \eqref{eq:gotthis} for all $i \in S$, we may reconstruct
\[ c_{i^*, \bullet} \cdot \rowof{W}{j}{i^*} \]
for all $j \in [0, t-1]$.  Because $\{ \rowof{W}{j}{i^*} \mid j \in [0,t-1]\}$ has dimension $t$ in $\F_q^t$, this is enough information to recover $c_{i^*, \bullet}$, as desired.

Now we show that 1 implies 2.  Suppose that we have a linear repair scheme for $i^*$ using $S$ with bandwidth $b$, so that for every $i \in S$ and every $j \in [0,t-1]$, there is some set $Q_{i,j} \subset \F_q^t$ so that
\begin{align*}
c_{i^*, j} &= \sum_{i \in S} \sum_{\bm{v} \in Q_{i,j}} \bm{v} \cdot c_{i, \bullet} \\
&= \sum_{i \in S} \Bigl( \sum_{\bm{v} \in Q_{i,j}} \bm{v} \Bigr) \cdot c_{i,\bullet} \\
&:= \sum_{i \in S} \bm{w}_{i,j} \cdot c_{i,\bullet}
\end{align*}
for every $C \in \cC$, where the final line defines the vectors $\bm{w}_{i,j} \in \F_q^t$.  
Moreover, we have
\[ \sum_{i \in S} \ \bigl| \bigcup_{j \in [0,t-1]} Q_{i,j} \bigr| := \sum_{i \in S} b_i \leq b \]
because $b_i := |\bigcup_{j \in [0,t-1]} Q_{i,j}|$ is the number of symbols of $\F_q$ returned by node $i$.
Now for $j \in [0,t-1]$, we define a repair matrix $W^{(j)}$ by
\[ \rowof{W}{j}{i} = \begin{cases} \bm{w}_{i,j} & i \in S \\ -\bm{e_j} & i = i^* \\ \bm{0} & \text{else} \end{cases} \]
where $\bm{e_j}$ denotes the $j^{th}$ standard basis vector in $\F_q^t$.
It is easily checked that the matrices $W^{(j)}$ are a valid set of repair matrices.  First, we see that by definition of $W^{(j)}$, we have $\langle W^{(j)}, C \rangle = 0$ for all $C \in \cC$ and $j \in [0,t-1]$.  Second,  
\[\{ \rowof{W}{j}{i^*} \mid j \in [0,t-1] \} = \{ -\bm{e_j} \mid j \in [0,t-1] \} \] 
is full rank.
Finally, for $i \neq i^*$,
\[\{ \rowof{W}{j}{i} \mid j \in [0,t-1] \} = \{ -\bm{w_{i,j}} \mid j \in [0,t-1] \} \]
where 
\[ \bm{w_{i,j}} = \sum_{\bm{v} \in Q_{i,j}} \bm{v} \]
In particular, all of these vectors live in the set $\bigcup_{j \in [0,t-1]} Q_{i,j}$, which has size $b_i$ as defined above.  
Thus, it has dimension at most $b_i$, and so 
\[ \sum_{i \in S} \dim( \{ \rowof{W}{j}{i} \mid j \in [0,t-1] \} ) \leq b \]
This completes the proof.


\end{document}